\documentclass[journal]{IEEEtran}
\usepackage{cite}
\usepackage{url}
\usepackage{amsmath}
\allowdisplaybreaks
\usepackage{amssymb}
\usepackage{enumitem}
\usepackage{bm}
\usepackage{algorithm,algorithmic}
\usepackage{mathtools}
\mathtoolsset{showonlyrefs}
\usepackage[caption=false,font=footnotesize,subrefformat=parens,labelformat=parens]{subfig}
\usepackage{refcount}
\usepackage{color}\definecolor{RED}{rgb}{1,0,0}\definecolor{BLUE}{rgb}{0,0,1}
\DeclareMathOperator*{\minimize}{minimize}
\DeclareMathOperator*{\maximize}{maximize}
\DeclareMathOperator*{\st}{subject\ to}
\DeclareMathOperator*\argmin{\text{argmin}}
\DeclareMathOperator*\argmax{\text{argmax}}

\def\ra{\rightarrow}

\def\ln{\text{ln}}

\newcommand{\mathletter}[1]{%
	\expandafter\newcommand\csname b#1\endcsname{\mathbb #1}
	\expandafter\newcommand\csname c#1\endcsname{\mathcal #1}
	\expandafter\newcommand\csname f#1\endcsname{\mathfrak #1}
	\expandafter\newcommand\csname til#1\endcsname{\widetilde #1}
	\expandafter\newcommand\csname ha#1\endcsname{\widehat #1}
	\expandafter\newcommand\csname bf#1\endcsname{\bf #1}
	\expandafter\newcommand\csname s#1\endcsname{\mathsf #1}
}%
\def\mathletters#1{\mathlettersB #1,,}
\def\mathlettersB#1,{\ifx,#1,\else\mathletter #1\expandafter\mathlettersB\fi}
\mathletters{A,B,C,D,E,F,G,H,I,J,K,L,M,N,O,P,Q,R,S,T,U,V,W,X,Y,Z}

\newcommand{\mathletterl}[1]{%
	\expandafter\providecommand\csname v#1\endcsname{\vec{#1}}
}%
\def\mathlettersl#1{\mathlettersC #1,,}
\def\mathlettersC#1,{\ifx,#1,\else\mathletterl #1\expandafter\mathlettersC\fi}
\mathlettersl{a,b,c,d,e,f,g,h,i,j,k,l,m,n,o,p,q,r,s,t,u,v,w,x,y,z}

\def \qed {\hfill \vrule height6pt width 6pt depth 0pt}
\def\bea{\begin{equation}\begin{alignedat}{-1}}
\def\ena{\end{alignedat}\end{equation}}
\def\bee{\begin{equation}}
\def\ene{\end{equation}}

\renewcommand{\vec}[1]{\mathbf{#1}}

\newtheorem{theo}{Theorem}
\newtheorem{lemma}{Lemma}
\newtheorem{assum}{Assumption}

\newtheorem{remark}{Remark}

\newenvironment{proof}{\begin{IEEEproof}}{\end{IEEEproof}}

\def\T{\mathsf{T}}

\def\bone{{\mathbf{1}}}

\renewcommand{\L}{L}
\begin{document}

\title{Distributed Dual Gradient Tracking for Resource Allocation in Unbalanced Networks}
\author{Jiaqi~Zhang,  Keyou~You, \IEEEmembership{Senior Member,~IEEE}, and Kai~Cai, \IEEEmembership{Senior Member,~IEEE}
	\thanks{* This work was  supported by the National Natural Science Foundation of China under Grant 61722308 and Dong Guan Innovative Research Team Program under Grant  2018607202007. ({\em Corresponding author: Keyou You}).}
	\thanks{J. Zhang and K. You are with the Department of Automation, and BNRist, Tsinghua University, Beijing 100084, China. E-mail: zjq16@mails.tsinghua.edu.cn, youky@tsinghua.edu.cn.}
	\thanks{K. Cai is with Department of Electrical and Information Engineering, Osaka City University, Osaka 558-8585, Japan. E-mail: kai.cai@eng.osaka-cu.ac.jp.}
}

\maketitle

\IEEEpeerreviewmaketitle

\begin{abstract}
	This paper proposes a distributed dual gradient tracking algorithm (DDGT) to solve resource allocation problems over an unbalanced network, where each node in the network holds a private cost function and computes the optimal resource by interacting only with its neighboring nodes. Our key idea is the novel use of the distributed push-pull gradient  algorithm (PPG) to solve the dual problem of the resource allocation problem. To study the convergence of the DDGT, we first establish the sublinear convergence rate of PPG for {\em non-convex} objective functions, which advances the existing results on PPG as they require the strong-convexity of objective functions.  Then we show that  the DDGT converges linearly for strongly convex and Lipschitz smooth cost functions, and sublinearly without the Lipschitz smoothness. Finally, experimental results suggest that DDGT outperforms existing algorithms.
\end{abstract}


\section{Introduction}
Distributed resource allocation problems  (DRAPs) are concerned with optimally allocating resources to multiple nodes that are connected via a directed peer-to-peer network. Each node is associated with a local private objective function to measure the cost of its allocated resource, and the global goal is to jointly minimize the total cost. The key feature of the DRAPs is that each node computes its optimal amount of resources by interacting {\em only} with its neighboring nodes in the network.  A typical application is the economic dispatch problem, where the local cost function is often quadratic \cite{yang2013consensus}. See \cite{ho1980class,xiao2006optimal,lakshmanan2008decentralized,zhao2018privacy} for other applications.

\subsection{Literature review}
Existing works on DRAPs can be categorized depending on whether the underlying network is {\em balanced} or not. A balanced network means that the ``amount" of information \emph{to} any node is equal to that \emph{from} this node, which is crucial to the  algorithm design. Most of early works on DRAPs focus on balanced networks and the recent interest is shifted to the unbalanced case.

The central-free algorithm (CFA) in \cite{ho1980class} is the first documented result on DRAPs in balanced networks where  at each iteration every node updates its decision variables using the weighted error between the gradient of its local objective function and those of its neighbors, and it can be accelerated  by designing an optimal weighting matrix \cite{xiao2006optimal}. It is proved that the CFA achieves a linear convergence rate for strongly convex and Lipschitz smooth cost functions. For time-varying networks,  the CFA is shown to converge sublinearly in the absence of strong convexity \cite{lakshmanan2008decentralized}. This rate is further improved in \cite{doan2017distributed} by optimizing its dependence on the number of nodes. In addition, there are also several ADMM-based methods that only work for balanced networks \cite{chang2014multi,chang2016proximal,aybat2016distributed}. By exploiting the mirror relationship between the distributed optimization and distributed resource allocation, several accelerated distributed  algorithms are proposed in  \cite{nedic2018improved,xu2018dual}. Moreover, \cite{liang2019distributed} and \cite{zhu2019distributed} study continuous-time algorithms for DRAPs by using the machinery of control theory.

For unbalanced networks, the algorithm design for DRAPs is much more complicated, which has been widely acknowledged in the distributed optimization literature \cite{xie2018distributed,cai2012average}. In this case, a consensus based algorithm that adopts the celebrated surplus idea \cite{cai2012average} is proposed in \cite{yang2013consensus} and \cite{xu2015fully}. However, their convergence results are only for quadratic cost functions where the analyses rely on the linear system theory.  The extension to general convex functions is performed in \cite{xu2017distributed} by adopting the nonnegative surplus method, at the expense of a slower convergence rate. The ADMM-based algorithms are developed in \cite{li2018admm,falsone2019tracking}, and algorithms that aim to handle communication delay in time-varying networks and perform event-triggered updates are respectively studied in \cite{yang2017distributed} and \cite{shi2018distributed}. We note that all the above-mentioned works \cite{yang2013consensus,xu2015fully,xu2017distributed,yang2017distributed,shi2018distributed,li2018admm,falsone2019tracking} do not provide explicit convergence rates for their algorithms. In contrast, the algorithm proposed in this work is proved to achieve a linear convergence rate for strongly convex and Lipschitz smooth cost functions, and has a sublinear convergence rate without the Lipschitz smoothness.

There are several recent works with convergence rate analyses of their algorithms over unbalanced networks. Most of them leverage the dual relationship between DRAPs and distributed optimization problems.  For example, the algorithms in  \cite{zhang2019distributed} and \cite{yuan2019stochastic} use stochastic gradients and diminishing stepsize to solve the dual problem of DRAPs, and thus their convergence rates are limited to an order of $O(\ln(k)/\sqrt{k})$ for Lipschitz smooth cost functions. \cite{yuan2019stochastic} also shows a rate of $O(\ln(k)/{k})$  if the cost function is strongly convex. An algorithm with linear  convergence rate is recently proposed in \cite{li2018convergence} for strongly convex and Lipschitz smooth cost functions. However, its convergence rate is unclear if either the strongly convexity or the Lipschitz smoothness is removed. In \cite{aybat2016distributed}, a push-sum-based algorithm is proposed by incorporating the alternating direction method of multipliers (ADMM). Although it can handle time-varying networks, the convergence rate is $O(1/k)$ even for strongly convex and Lipschitz smooth functions.

\subsection{Our contributions}

In this work, we propose a distributed dual gradient tracking algorithm (DDGT) to solve DRAPs over unbalanced networks. The DDGT exploits the duality of DRAPs and distributed optimization problems, and takes advantage of the distributed push-pull gradient algorithm (PPG) \cite{pu2018push}, which is also called $\mathcal{AB}$ algorithm in \cite{xin2018linear}. If the  cost function is strongly convex and Lipschitz smooth,  we show that the DDGT converges at a linear rate $O(\lambda^k),\lambda\in(0,1)$. If the Lipschitz smoothness  is not satisfied, we show the convergence of the DDGT and establish an  convergence rate $O(1/k)$. To our best knowledge, these convergence results are only reported for \emph{undirected} or \emph{balanced} networks in \cite{nedic2018improved}. Although  a distributed algorithm for directed networks is also proposed in  \cite{nedic2018improved}, there is no convergence analysis. The advantages of the DDGT over existing algorithms are also validated by numerical experiments.

To characterize the sublinear convergence of the DDGT, we first show that PPG  converges sublinearly to a stationary point even for {\em non-convex} objective functions. Clearly, this advances existing works \cite{pu2018push,xin2018linear,saadatniaki2018optimization} as their convergence results are only for strongly-convex objective functions. In fact, the convergence proofs for  PPG in \cite{pu2018push,xin2018linear,saadatniaki2018optimization} require constructing a complicated 3-dimensional matrix and then derive the linear convergence rate $O(\lambda^k)$ where $\lambda\in(0,1)$ is the spectral radius of this matrix. This approach is no longer applicable since a linear convergence rate is usually not attainable for general non-convex functions \cite{nesterov2013introductory} and hence the spectral radius of such a matrix cannot be strictly less than one.

\subsection{Paper organization and notations}

The rest of this paper is organized as follows. In Section \ref{sec2}, we formulate the constrained DRAPs with some standard assumptions. Section \ref{sec3} firstly derives the dual problem of DRAPs which is amenable to distributed optimization, and then introduces the PPG. The DDGT is then obtained by applying PPG to the dual problem and improving the initialization. In Section \ref{sec4}, the convergence result of the DDGT is derived by establishing the convergence of PPG for non-convex objective functions. Section \ref{sec5} performs numerical experiments to validate the effectiveness of the DDGT. Finally, we draw conclusive remarks in Section \ref{sec6}.

We use a lowercase $x$, bold letter $\vx$ and uppercase $X$ to denote a scalar, vector, and matrix, respectively. $\vx^\T$ denotes the transpose of the vector $\vx$. $[X]_{ij}$ denotes the element in the $i$-th row and $j$-th column of the matrix $X$. For vectors we use $\|\cdot\|$ to denote the $l_2$-norm. For matrices we use $\|\cdot\|$ and $\|\cdot\|_F$ to denote respectively the spectral norm and the Frobenius norm. $|\cX|$ denotes the cardinality of set $\cX$. $\bR^n$ denotes the set of $n$-dimensional real vectors. $\bone$ denotes the vector with all ones, the dimension of which depends on the context. $\nabla f(\vx)$ denotes the gradient of a differentiable function $f$ at $\vx$. We say a nonnegative matrix $X$ is row-stochastic if $X\bone=\bone$, and column-stochastic if $X^\T$ is row-stochastic. $O(\cdot)$ denotes the big-O notation.

\section{Problem formulation}\label{sec2}
Consider the distributed resource allocation problems (DRAPs) with $n$ nodes, where each node $i$ has a local private cost function $F_i:\bR^m\ra\bR$. The goal is to solve the following optimization problem in a distributed manner:
\bea\label{obj}
	&\minimize_{\vw_1,\cdots,\vw_n\in\bR^m} &&\sum_{i=1}^n F_i(\vw_i)\\
	&\st &&\vw_i\in\cW_i,\ \sum_{i=1}^{n}\vw_i = \sum_{i=1}^{n}\vd_i
\ena
where $\vw_i\in \bR^m$ is the local decision vector of node $i$, representing the resources allocated to $i$. $\cW_i$ is a local convex and closed constraint set. $\vd_i$ denotes the resource demand of node $i$. Both $\cW_i$ and $\vd_i$ are only known to node $i$. Let $\vd\triangleq\sum_{i=1}^{n}\vd_i$, then $\sum_{i=1}^{n}\vw_i =\vd$ represents the constraint on total available resources, showing the coupling among nodes.
\begin{remark}
	Problem \eqref{obj} covers many forms of DRAPs considered in the literature. For example, the standard local constraint  $\cW_i=[\underline{w}_i, \overline{w}_i]$ for some constants $\underline{w}_i$ and $\overline{w}_i$ is a one-dimensional special case of \eqref{obj}, see e.g. \cite{xu2017distributed,xu2015fully,yang2013consensus,yang2017distributed,li2018convergence}. Moreover, the coupling constraint can be given in a weighted form $\sum_{i=1}^{n}A_i\vw_i = \vd$, which can be transformed into \eqref{obj} by defining a new variable $\vw'_i=A_i\vw_i$ and a local constraint set $\cW'_i=\{A_i\vw_i|\vw_i\in\cW_i\}$. In addition, many works only consider quadratic cost functions\cite{xu2015fully,yang2013consensus}.
\end{remark}

Solving \eqref{obj} in a distributed manner means that each node can only communicate and exchange information with a subset of nodes via a communication network, which is modeled by a directed graph $\cG=(\cV,\cE)$. Here $\cV=\{1,\cdots,n\}$ denotes the set of nodes, $\cE\subseteq\cV\times\cV$ denotes the set of edges, and $(i,j)\in\cE$ if node $i$ can send information to node $j$. Note that $(i,j)\in\cE$ does not necessarily imply that $(j,i)\in\cE$. Define $\cN_i^{\text{in}}=\{j|(j,i)\in\cE\}\cup\{i\}$ and $\cN_i^{\text{out}}=\{j|(i,j)\in\cE\}\cup\{i\}$ as the set of in-neighbors and out-neighbors of node $i$, respectively. That is, node $i$ can only receive messages from its in-neighbors and send messages to its out-neighbors.  Let $a_{ij}>0$ be the weight associated to edge $(j,i)\in\cE$. $\cG$ is {\em balanced} if $\sum_{j\in \cN_i^{\text{in}}} a_{ij}=\sum_{j\in\cN_i^{\text{out}}} a_{ji}$ for all $i\in\cV$. Note that balancedness is a relatively strong condition, since it can be difficult or even impossible to find weights satisfying it for a general directed graph \cite{gharesifard2010does}. 

The following assumptions are made throughout the paper.
\begin{assum}[Strong convexity and Slater's condition]\label{assum1}
	\begin{enumerate}
		\item	The local cost function $F_i$ is $\mu$-strongly convex for all $i\in\cV$, i.e., for any $\vw,\vw'\in\bR^m$ and $\theta\in[0,1]$,
		      \bea
			      &F_i(\theta\vw+(1-\theta)\vw')\\
			      &\leq\theta F_i(\vw)+(1-\theta)F_i(\vw')-\frac{\mu}{2}\theta(1-\theta)\|\vw-\vw'\|^2.
		      \ena
		\item The constraint $\sum_{i=1}^{n}\vw_i = \vd$ is satisfied for some point in the relative interior of the Cartesian product $\cW:=\cW_1\times\cdots\times\cW_n$.
	\end{enumerate}
\end{assum}

\begin{assum}[Strongly connected network]\label{assum2}
	$\cG$ is strongly connected, i.e., there exists a directed path from any node $i$ to any node $j$.
\end{assum}

Assumption \ref{assum1} is common in the literature. Note that we do not assume the differentiability of $F_i$. Under Assumption \ref{assum1}, the optimal point of \eqref{obj} is unique. Let $F^\star$ and $\vw_i^\star,i\in\cV$ denote respectively its optimal value and optimal point, i.e., $F^\star=\sum_{i=1}^{n}F_i(\vw_i^\star)$.  Assumption \ref{assum2} is also common and necessary for the information mixing over a network.

\section{The Distributed Dual Gradient Tracking Algorithm}\label{sec3}

This section introduces our distributed dual gradient tracking algorithm (DDGT)  to solve \eqref{obj} over an unbalanced network. We start with the dual problem of \eqref{obj} and transform it as a form of distributed optimization. Then, the DDGT is obtained by using the push-pull gradient method (PPG \cite{pu2018push,xin2018linear}) on the dual problem, which is an efficient distributed optimization algorithm over \emph{unbalanced} networks. 

\subsection{The dual problem of \eqref{obj} and PPG }\label{sec3b}

Define the Lagrange function of \eqref{obj} as
\bee\label{lagrange}
L(W,\vx)=\sum_{i=1}^n F_i(\vw_i)+\vx^\T(\sum_{i=1}^{n}\vw_i-\vd)
\ene
where $W=[\vw_1,\cdots,\vw_n]\in\bR^{m\times n}$ and $\vx$ is the Lagrange multiplier. Then, the dual problem of \eqref{obj} is given by
\bee\label{dual}
	\maximize _{\vx\in\bR^m}\ \inf_{W\in\cW} L(W,\vx).
\ene

Under Assumption \ref{assum1}, the strong duality holds \cite{boyd2004convex}, \cite[Exercise 5.2.2]{bertsekas2016nonlinear}. The objective function in \eqref{dual} is written as
\bea\label{dual_obj}
	\inf_{W\in\cW} L(W,\vx)&=\inf_{W\in\cW}\sum_{i=1}^{n}(F_i(\vw_i)+\vx^\T \vw_i)-\vx^\T \vd\\
	&=\sum_{i=1}^{n}\inf_{\vw_i\in\cW_i}\{F_i(\vw_i)+\vx^\T \vw_i\}-\vx^\T \vd\\
	&=\sum_{i=1}^{n}-F_i^*(-\vx)-\vx^\T\vd
\ena
where
\bee\label{eq1_sec3}
F_i^*(\vx)\triangleq \sup_{\vw_i\in\cW_i}\{\vw_i^\T\vx-F_i(\vw_i)\}
\ene 
is the convex conjugate function corresponding to the pair $(F_i,\cW_i)$ \cite[Section 5.4]{bertsekas2016nonlinear}. Thus, the dual problem \eqref{dual} can be rewritten as a convex optimization problem
\bee\label{dual_obj2}
	\minimize_{\vx\in\bR^m}\ f(\vx)\triangleq\sum_{i=1}^{n} f_i(\vx),\ f_i(\vx)\triangleq F_i^*(-\vx)+\vx^\T\vd_i
\ene
or equivalently,
\bea\label{obj2}
	&\minimize_{\vx_1,\cdots,\vx_n\in\bR^m} &&\sum_{i=1}^{n}f_i(\vx_i)\\
	&\st &&\vx_1=\cdot\cdot\cdot=\vx_n.
\ena

Recall that strong duality holds, and therefore problem \eqref{obj2} is equivalent to problem $\eqref{obj}$ in the sense that the optimal value of \eqref{obj2} is $f^\star=-F^\star$ and the optimal point $\vx_1^\star=\cdots=\vx_n^\star=\vx^\star$ of \eqref{obj2} satisfies $F_i(\vw_i^\star)+F_i^\ast(-\vx^\star)=-(\vw_i^\star)^\T\vx^\star$. Hence, we can simply focus on solving the dual problem \eqref{obj2}.

The strong convexity of $F_i$ implies that $F_i^*$ is differentiable with Lipschitz continuous gradients \cite{boyd2004convex}, and the supremum in \eqref{eq1_sec3} is attainable. By Danskin's theorem \cite{bertsekas2016nonlinear}, the gradient of $F_i^*$ is given by
$
	\nabla F_i^*(\vx)=\argmax_{\vw\in\cW_i}\{\vx^\T\vw-F_i(\vw)\}.
$
Thus, it follows from \eqref{dual_obj2} that
\bea\label{eq3_sec1}
	\nabla f_i(\vx)&=-\nabla F_i^*(-\vx)+\vd_i\\
	&=-\argmin_{\vw\in\cW_i}\{\vx^\T\vw+F_i(\vw)\}+\vd_i.
\ena

The dual form \eqref{obj2} allows us to take advantage of recent advances in distributed optimization  to solve DRAPs over \emph{unbalanced} networks. For example, distributed algorithms are proposed in \cite[gradient-push]{nedic2015distributed},  \cite[Push-DIGing]{nedic2017achieving}, \cite[ExtraPush]{zeng2017extrapush}, \cite[DEXTRA]{xi2017dextra}, \cite{mai2016distributed} to solve \eqref{obj2} over general directed and unbalanced graphs. Asynchronous algorithms are also studied in \cite{zhang2019asyspa,zhang2019asynchronous,zhao2015asynchronous,wu2018decentralized}. In particular, \cite{pu2018push} and \cite{xin2018linear} propose  PPG algorithm (or called $\cA\cB$ in \cite{xin2018linear}) by using the idea of gradient tracking, which achieves a linear convergence rate if the objective function $f_i$ is strongly convex and Lipschitz smooth for all $i$. Moreover, PPG has an empirically faster convergence speed than its competitors (e.g. \cite{nedic2017achieving}), and its linear update rule is an advantage for implementation. The compact form of PPG is given as
\bea\label{ab}
	\vx_{k+1}^{(i)}&=\sum_{j\in\cN_i^{\text{in}}}a_{ij}(\vx_{k}^{(j)}-\alpha \vy_k^{(j)})\\
	\vy_{k+1}^{(i)}&=\sum_{j\in\cN_i^{\text{in}}}b_{ij}\vy_{k}^{(j)}+\nabla f_i(\vx_{k+1}^{(i)})-\nabla f_i(\vx_{k}^{(i)})
\ena
where $a_{ij}> 0$ for any $j\in\cN_i^{\text{in}}$ and $\sum_{j\in\cN_i^{\text{in}}} a_{ij}=1$, $b_{ij}>0$ for any $i\in\cN_j^{\text{out}}$ and $\sum_{i\in\cN_j^{\text{out}}} b_{ij}=1$, $\alpha$ is a positive stepsize, and $\vx_0^{(i)}$ and $\vy_0^{(i)}$ are initialized such that $\vy_0^{(i)}=\nabla f_i(\vx_0^{(i)}),\forall i\in \cV$. Intuitively, the update for $\vy_k^{(i)}$ aims to asymptotically track the global gradient $\nabla f(\bar \vx_k)$ and the update for $\vx_k^{(i)}$ enforces it to converge to $\bar\vx_k$  while performing an inexact gradient descent step, where $\bar\vx_k=\frac{1}{n}\sum_{i=1}^{n}\vx_k^{(i)}$ is the mean of nodes' states. We refer interested readers to \cite{pu2018push,xin2018linear} for more discussions on PPG.

\subsection{The DDGT algorithm}

We are ready to present the DDGT algorithm. Plugging the gradient \eqref{eq3_sec1} into \eqref{ab} and noticing that the $\vd_i$ term is cancelled in $\nabla f_i(\vx_{k+1}^{(i)})-\nabla f_i(\vx_{k}^{(i)})$, we have
\begin{subequations}\label{alg}
	\noeqref{eq1_alg,eq2_alg,eq3_alg,eq2b_alg,eq2c_alg}
	\begin{align}
		\overline{\vw}_{k+1}^{(i)} & =\sum_{j\in\cN_i^{\text{in}}} a_{ij}(\overline{\vw}_{k}^{(j)} + \alpha \vs_k^{(j)}),\label{eq1_alg} \\
		\vw_{k+1}^{(i)}       & =\argmin_{\vw\in\cW_i}\{F_i(\vw)-\vw^\T\overline\vw_{k+1}^{(i)}\},\label{eq2_alg}                   \\
		\vs_{k+1}^{(i)}       & =\sum_{j\in\cN_i^{\text{in}}} b_{ij}\vs_k^{(j)}-(\vw_{k+1}^{(i)}-\vw_k^{(i)}).\label{eq3_alg}
	\end{align}
\end{subequations}
where notations have been changed to keep consistency with the primal problem \eqref{obj}, i.e., $\vx_{k}^{(i)}=-\overline\vw_{k}^{(i)}$ and $\vy_{k}^{(i)}=\vs_{k}^{(i)}$. 

The DDGT is summarized in Algorithm \ref{alg_doba} and we now elaborate on it. After initialization, each node $i$ iteratively updates three vectors $\overline{\vw}_{k}^{(i)}, {\vw}_{k}^{(i)}$ and ${\vs}_{k}^{(i)}$. In particular, at each iteration node $i$ receives $\widetilde{\vw}_{k}^{(j)}:=\overline{\vw}_{k}^{(j)}+\alpha\vs_k^{(j)}$ and $\widetilde\vs_k^{(ji)}:=b_{ij}\vs_k^{(j)}$ from each of its in-neighbors $j$, and updates $\overline{\vw}_{k+1}^{(i)}$ according to \eqref{eq1_alg}, where $a_{ij}$ is positive for any $j\in\cN_i^{\text{in}}$ such that $\sum_{j\in\cN_i^{\text{in}}} a_{ij}=1$ as with \eqref{ab}, and $\alpha$ is a positive stepsize. The update of ${\vs}_{k}^{(i)}$ in \eqref{eq3_alg} is similar, where $b_{ij}>0$ for any $i\in\cN_j^{\text{out}}$ and $\sum_{i\in\cN_j^{\text{out}}} b_{ij}=1$. This process repeats until terminated.
We set $a_{ij}=b_{ij}=0$ for any $(j,i)\notin\cE$ for convenience. Define two matrices $[A]_{ij}=a_{ij}$ and $[B]_{ij}=b_{ij}$, then $A$ is a row-stochastic matrix and $B$ is a column-stochastic matrix. Clearly, the directed network associated with $A$ and $B$ can be unbalanced.

\begin{remark}
	In practice, one can simply set $a_{ij}={|\cN_i^{\text{in}}|^{-1}}$ and $b_{ij}={|\cN_j^{\text{out}}|^{-1}}$, and then all conditions are satisfied. Note that this setting requires each node to know the number of its in-neighbors and out-neighbors, which is common in the literature of distributed optimization over directed networks \cite{nedic2015distributed,nedic2017achieving,zeng2017extrapush,xi2017dextra}.
\end{remark}

Notably, the initialization for DDGT exploits the structure of the DRAPs and improves that of PPG. By PPG,  $\vw_0^{(i)}$ and $\vs_0^{(i)}$ should be exactly set as $\vw_0^{(i)}=\widetilde\vw_i^\star$ and $\vs_0^{(i)}=\vd_i-\widetilde\vw_i^\star$, where $\widetilde\vw_i^\star=\argmin_{\vw\in\cW_i}F_i(\vw)$ is a local minimizer. In DDGT, the computation of $\widetilde\vw_i^\star$ is actually not necessary since the update without $\widetilde\vw_i^\star$ in $\vw_0^{(i)}$ and $\vs_0^{(i)}$ and the update with it become equivalent after the first iteration due to the special form of $\nabla f_i(\vx)$. Clearly, the former is simpler and is adopted in DDGT.

The update of $\vw_{k}^{(i)}$ in \eqref{eq2_alg} requires finding an optimal point of an auxiliary \emph{local} optimization problem, which can be obtained by standard algorithms, e.g., projected (sub)gradient method or Newton's method, and can even be given in an explicit form for some special cases. Note that solving sub-problems per iteration is common in many duality-based optimization algorithms, including the dual ascent method and proximal method \cite{bertsekas2015convex}.

\begin{remark}
Consider two special cases. The first one is that the local constraint set $\cW_i=\bR^m$ and $F_i$ is differentiable as in \cite{lakshmanan2008decentralized}. Then,  \eqref{eq2_alg} becomes 
\bee\label{eq2b_alg}
	\tag{9b$'$}
	\vw_{k+1}^{(i)}=\nabla^{-1}F_i(\overline\vw_{k+1}^{(i)})
\ene
where $\nabla^{-1}F_i$ denotes the inverse function of $\nabla F_i$, i.e., $\nabla^{-1}F_i(\nabla F_i(\vx))=\vx$ for any $\vx\in\bR^m$. 

The second case is that the decision variable is a scalar, $\cW_i$ is an interval $[\underline{w}_i,\overline{w}_i]$, and $F_i$ is differentiable as in  \cite{xu2017distributed,yang2017distributed,yang2013consensus}. Then, \eqref{eq2_alg} becomes
\bee\label{eq2c_alg}
	\tag{9b$''$}
	\vw_{k+1}^{(i)}=\left\{\begin{array}{ll}
		\overline{w}_i,                         & \text{if }\nabla^{-1}F(\overline\vw_{k+1}^{(i)})>\overline{w}_i       \\
		\underline{w}_i,                   & \text{if }\nabla^{-1}F(\overline\vw_{k+1}^{(i)})<\underline{w}_i \\
		\nabla^{-1}F(\overline\vw_{k+1}^{(i)}), & \text{otherwise}
	\end{array}\right.
\ene
which is in fact adopted in \cite{yang2013consensus,xu2017distributed,yang2017distributed}.  Hence, \eqref{eq2_alg} can be seen as an extension of their methods. \qed
\end{remark}

An interesting feature of DDGT lies in the way to handle the coupling constraint $\sum_{i=1}^{n}\vw_k^{(i)} = \vd$. Notice that DDGT is simply initialized such that $\vw_0^{(i)}=0,\forall i\in\cV$ and $\sum_{i=1}^{n}\vs_0^{(i)}=\vd$. By summing \eqref{eq3_alg} over $i=1,\cdots,n$, we obtain that $\sum_{i=1}^{n}(\vw_k^{(i)}+\vs_k^{(i)})=\sum_{i=1}^{n}(\vw_0^{(i)}+\vs_0^{(i)})=\vd$. Thus, if $\vs_{k}^{(i)}$ converges to 0, the constraint is satisfied asymptotically, which is essential to the convergence proof of the DDGT.

\begin{algorithm}[t!]
	\makeatletter
	\renewcommand\footnoterule{%
		\kern-3\p@
		\hrule\@width.4\columnwidth
		\kern2.6\p@}
	\makeatother
	\begin{minipage}{\linewidth}
	\caption{The Distributed Dual Gradient Tracking Algorithm (DDGT) --- from the view of node $i$}\label{alg_doba}
	\label{algorithm}
	\begin{itemize}[leftmargin=*]
		\item{\bf Initialization:} Let $\overline{\vw}_{0}^{(i)}=0$, $\vw_0^{(i)}=0$, $\vs_{0}^{(i)}=\vd_i$.\footnote{\label{footnote1}If only the total resource demand $\vd$ is known to all nodes, then we can simply set $\vs_0^{(i)}=\frac{1}{n}\vd$, which can be done in a distributed manner \cite{xu2017distributed}.}
		\item{\bf For $k=0,1,\cdots,K$,} {\bf repeat}
		      \begin{enumerate}[leftmargin=*]
			      \renewcommand{\labelenumi}{\theenumi:}
			      \item Receive $\widetilde{\vw}_{k}^{(j)}:=\overline{\vw}_{k}^{(j)}+\alpha\vs_k^{(j)}$ and $\widetilde\vs_k^{(ji)}:=b_{ij}\vs_k^{(j)}$ from its in-neighbor $j$.
			      \item Compute $\overline{\vw}_{k+1}^{(i)}$, ${\vw}_{k+1}^{(i)}$ and $\vs_{k+1}^{(i)}$ as \eqref{alg}.
			      \item Broadcast $\widetilde{\vw}_{k+1}^{(i)}$ and $\widetilde \vs_{k+1}^{(i)}$ to each of out-neighbors.
		      \end{enumerate}
		\item{{\bf Return} $\vw_K^{(i)}$.}
	\end{itemize}
\end{minipage}
\end{algorithm}

By strong duality, the convergence of DDGT can be established by showing the convergence of PPG. However, existing results, e.g.,\cite{pu2018push,xin2018linear,pu2018distributed,saadatniaki2018optimization}  for the convergence of PPG are established only if $f_i$ is strongly convex and Lipschitz smooth. Note that $f_i$ in \eqref{obj2} is often \emph{not} strongly convex due to the introduction of convex conjugate function $F_i^*$, though $F_i$ in \eqref{obj} is strongly convex \cite{boyd2004convex}. This is indeed the case if $F_i$ includes exponential term \cite{zhao2017distributed} or logarithmic term \cite{zheng2014optimal}. Without Lipschitz smoothness for $F_i$, we can only obtain that $f_i$ is differentiable and $\frac{1}{\mu}$-Lipschitz smooth \cite[Theorem 4.2.1]{Hiriart-Urruty1993}, i.e.,
\bee\label{lsmooth}
	\|\nabla f_i(\vx)-\nabla f_i(\vy)\|\leq \frac{1}{\mu}\|\vx-\vy\|, \forall i\in\cV,\vx,\vy\in\bR^n.
\ene

Thus, we still need to prove the convergence of PPG for non-strongly convex objective functions $f_i$. Particularly, a crucial step in the convergence proof of PPG in \cite{pu2018push,xin2018linear} uses a complicated 3-dimensional matrix whose spectral radius is strictly less than one for a sufficiently small stepsize. Then, PPG converges at a linear rate. This does not work here since the spectral radius of such a matrix cannot be strictly less than one if $f_i$ is not strongly convex. In fact, we cannot expect a linear convergence rate for the non-strongly convex case \cite{nesterov2013introductory}.

Next, we shall prove that PPG converges to a stationary point at a rate of $O(1/k)$ even for {\em non-convex} objective functions, based on which we  show the convergence and evaluate the convergence rate of DDGT.

\section{Convergence Analysis}\label{sec4}

In this section, we first establish the convergence result of PPG in \eqref{ab}  for non-convex $f_i$, which is of independent interest as the existing results on PPG only apply to the strongly convex case. Then, we show  the convergence of the DDGT and evaluate the convergence rate for a special case.

\subsection[convergence of AB]{Convergence analysis of PPG without convexity}

Consider PPG given in \eqref{ab}. With a slight abuse of notation, let $f_i$ be a general differentiable function in the rest of this subsection. Denote
\bea\label{eq1_sec4}
	X_k&=[\vx_k^{(1)},\cdots,\vx_k^{(n)}]^\T\in\bR^{n\times m}\\
	Y_k&=[\vy_k^{(1)},\cdots,\vy_k^{(n)}]^\T\in\bR^{n\times m}\\
	\nabla\vf_k&=[\nabla f_1(\vx_k^{(1)}),\cdots,\nabla f_n(\vx_k^{(n)})]^\T\in\bR^{n\times m}
\ena
and
\bea
[A]_{ij}=\left\{
\arraycolsep=3pt
\begin{array}{ll}
	a_{ij}, & \text{if }(j,i)\in\cE \\
	0,      & \text{otherwise,}
\end{array}
\right.\
\arraycolsep=3pt
[B]_{ij}=\left\{
\begin{array}{ll}
	b_{ij}, & \text{if }(j,i)\in\cE \\
	0,      & \text{otherwise.}
\end{array}
\right.
\ena
Note that $A$ is row-stochastic and $B$ is column-stochastic. The starting points of all nodes are set to the same point $\vx_0$ for simplicity.

Then, \eqref{ab} can be written in the following compact form
\begin{subequations}
	\label{ppg}
	\noeqref{ppga,ppgb}
	\begin{align}
		X_{k+1} & = A(X_k - \alpha Y_k)\label{ppga}                      \\
		Y_{k+1} & = BY_k + \nabla \vf_{k+1} - \nabla \vf_{k}\label{ppgb}
	\end{align}
\end{subequations}

The convergence result of PPG for non-strongly convex or even non-convex functions are given in the following result. 
\begin{theo}[Convergence of PPG without convexity]\label{theo1}
	Suppose Assumption \ref{assum2} holds and $f_i, i\in\cV$ in \eqref{obj2} is differentiable and $L$-Lipschitz smooth (c.f. \eqref{lsmooth}). If the stepsize $\alpha$ is sufficiently small, i.e., $\alpha$ satisfies \eqref{eq_alpha1} and \eqref{eq_alpha}, then $\{\vx_k^{(i)}\},i\in\cV$ generated by  \eqref{ab} satisfies that
	\bea\label{eq1_theo1}
		&\frac{1}{k}\sum_{t=1}^{k} \|\nabla f(\bar\vx_{t})\|^2\leq \frac{f(\vx_0)-f^\star}{\gamma k}+\frac{3L\alpha^2(L^2c_0^2+ c_2^2)}{\gamma(1-\theta)^2k}\\
		&\quad+\frac{\alpha(\sqrt{n}Lc_0+c_2)(1+\sum_{t=1}^{k_0}\|\nabla f(\bar\vx_{t})\|^2)}{\gamma(1-\theta)^2k}
	\ena
	where $\bar\vx_k=\sum_{i=1}^{n}\pi_A^{(i)}\vx_k^{(i)}$, $\pi_A$ is the normalized left Perron vector of $A$, and $\theta,c_0,c_2,\gamma, k_0$ are positive constants given in \eqref{eq4_s3}, \eqref{eq6_s3}, \eqref{eq_gamma}, \eqref{eq_k0} of Appendix, respectively. 
	
	Moreover, it holds that
	\bee\label{eq2_theo1}
		\frac{1}{k}\sum_{t=1}^{k} \|X_t-\bone\bar\vx_t^\T\|_F^2\leq\frac{2c_0^2}{(1-\theta)^2k}+\frac{c_1^2\alpha^2}{k}\sum_{t=1}^{k} \|\nabla f(\bar\vx_{t})\|^2
	\ene
	and if $f$ is convex, $f(\bar\vx_k)$ converges to  $f^\star$.
\end{theo}

The proof of Theorem \ref{theo1} is deferred to the Appendix. Theorem \ref{theo1} shows that PPG converges to a stationary point of $f$ at a rate of $O(1/k)$ for non-convex functions. The order of convergence rate is consistent with the centralized gradient descent algorithm \cite{bertsekas2015convex}. Generally, the network size $n$ affects the convergence rate in a complicated way since it closely relates to the network topology and the two weighting matrices $A$ and $B$.  If $\sigma_\sA,\sigma_\sB,\delta_{\sA \sF}$ and $\delta_{\sB \sF}$ in Lemmas 2 and 3 of Appendix do not vary with $n$, which holds, e.g., by setting $A=B$ in some undirected graphs such as complete graphs and star graphs, then it follows from \eqref{eq4_s3}, \eqref{eq_alpha1}, \eqref{eq6_s3} and \eqref{eq_gamma} that $\theta=O(1)$, $\alpha=O(1/\sqrt{n})$, $c_0\approx O(\sqrt{n})$ and $\gamma\approx O(\alpha)$. Then, \eqref{eq1_theo1} ensures a convergence rate $O(n/k)$, which is reasonable since the Lipschitz constant $L$ is defined in terms of \emph{local} objective functions, and the \emph{global} Lipschitz constant generally increases linearly with $n$, implying a convergence rate $O(n/k)$ even for the centralized gradient descent method \cite[Section 6.1]{bertsekas2015convex}.
\subsection{Convergence of the DDGT}

We now establish the convergence and quantify the convergence rate of the DDGT.
\begin{theo}[Convergence of the DDGT]\label{theo2}
	Suppose Assumptions \ref{assum1} and \ref{assum2} hold. If the stepsize $\alpha>0$ is  smaller than an upper bound  given in \eqref{eq_alpha1} and \eqref{eq_alpha} with $L$ replaced by $1/\mu$, 
 then $\{\vw_{k}^{(i)}\},i\in\cV$ in Algorithm \ref{algorithm}  converges to an optimal point of \eqref{obj}, i.e., $\lim_{k\ra\infty}\vw_{k}^{(i)}=\vw_i^\star,\forall i\in\cV$.
	\end{theo}
\begin{proof}
	Under Assumption \ref{assum1}, the strong duality holds between the original problem \eqref{obj} and its dual problem \eqref{obj2}. Recall the relation between the DDGT \eqref{alg} and PPG \eqref{ab}. We obtain that $f(\bar\vx_k)$ converges to $f^\star$ by the convexity of the dual problem and Theorem \ref{theo1}, and $f^\star=-F^\star=-L(W^\star,\vx)$ for any $x\in\bR^m$. Moreover, 
	\bea\label{eq_theo2}
	f(\bar\vx_k)-f^\star&=L(W^\star,\bar\vx_k)-\inf_{W\in\cW}L(W,\bar\vx_k)\\
	&= L(W^\star,\bar\vx_k)-L(W_k,\bar\vx_k)\\
	&\geq \partial_W L(W_k,\bar\vx_k)^\T(W^\star-W_k)+\frac{\mu}{2}\|W_k-W^\star\|_F^2\\
	&\geq \frac{\mu}{2}\|W_k-W^\star\|_F^2
	\ena
	where $L(W,\vx)$ is the Lagrange function in \eqref{lagrange}, $W^\star=[\vw_1^\star,\cdots,\vw_n^\star]$ and $W_k=[\vw_k^{(i)},\cdots,\vw_k^{(i)}]$. The first inequality follows from the strong convexity of $F$ by Assumption \ref{assum1} and the second inequality uses the first-order necessary condition for a constrained minimization problem. The convergence of $\vw_k^{(i)}$ is obtained immediately from \eqref{eq_theo2}. 
	
	The stepsize condition follows from Theorem \ref{theo1} and the Lipschitz smoothness  of the dual function (c.f. \eqref{lsmooth}).
\end{proof}
\begin{remark}
	We note that it is possible to extend the DDGT to time-varying networks \cite{xu2017distributed}, since the convergence of the DDGT essentially depends on that of PPG, and a recent work \cite{saadatniaki2018optimization} shows the feasibility of PPG over time-varying networks for strongly convex functions.
\end{remark}
\subsection{Convergence rate of the DDGT}
As in \cite{lakshmanan2008decentralized} and \cite{nedic2018improved}, this subsection focuses on the special case that $\cW_i=\bR^m$ and $F_i$ is differentiable for all $i\in\cV$ for the convergence rate characterization, since the constrained case involves more complicated concepts and notations such as subdifferential.  

Under Assumption \ref{assum1}, it follows from \cite{boyd2004convex} that the Karush-Kuhn-Tucker (KKT) condition of \eqref{obj} 
\begin{subequations}\label{eq2_sec1}
	\noeqref{eq2a_sec1,eq2b_sec1}
	\begin{align}
		 & \nabla F_1(\vw_1^\star)=\cdots=\nabla F_n(\vw_n^\star)\label{eq2a_sec1}, \\
		 & \sum_{i=1}^{n}\vw_i^\star=\vd\label{eq2b_sec1}
	\end{align}
\end{subequations}
 is a necessary and sufficient condition for optimality. The convergence rate of the DDGT is in terms of \eqref{eq2_sec1}.

\begin{theo}[Convergence rate of the DDGT]\label{theo3}
	Suppose that $\cW_i=\bR^m$, $F_i$ is differentiable for all $i$, and the conditions in Theorem \ref{theo2} are satisfied. Let $\overline{\nabla F_k}=\frac{1}{n}\sum_{i=1}^{n}\nabla F_i(\vw_k^{(i)})$, then $\{\vw_{k}^{(i)}\}$ generated by the DDGT satisfies that
	\bea
	&\frac{1}{k}\sum_{t=1}^{k}\Big(\sum_{i=1}^{n}\|\nabla F_i(\vw_{t}^{(i)})-\overline{\nabla F_t}\|^2+\|\sum_{i=1}^{n}\vw_t^{(i)}-\vd\|^2\Big)\\
	&\leq \frac{2(f(\vx_0)-f^\star)}{\gamma k}+\frac{6L\alpha^2(L^2c_0^2+ c_2^2)}{\gamma(1-\theta)k}+\frac{4nc_0^2(\mu^2+1)}{\mu^2(1-\theta)k}+\\
	&\quad+\frac{2\alpha(\sqrt{n}Lc_0+c_2)(1+\sum_{t=1}^{k_0}\|\nabla f(\bar\vx_{t})\|^2)}{\gamma(1-\theta)k}+O(\frac{1}{k^2})
	\ena
	where all the constants are defined in Theorem \ref{theo1}. 
	
	Moreover, if $F_i, i\in\cV$ has Lipschitz continuous gradients, then $\sum_{i=1}^{n}\|\vw_{k}^{(i)}-\vw_i^\star\|^2$ converges linearly, i.e., $\sum_{i=1}^{n}\|\vw_{k}^{(i)}-\vw_i^\star\|^2\leq O(\lambda^k)$ for some $\lambda\in(0,1)$.
\end{theo}

\begin{proof}
	Since $\cW_i=\bR^m$, it follows from \eqref{eq2_alg} that $\vx_{k+1}^{(i)}=-\nabla F_i(\vw_{k+1}^{(i)})$. Thus,
	\bea\label{eq1_thm3}
	&\sum_{i=1}^{n}\|\nabla F_i(\vw_{k}^{(i)})-\overline{\nabla F_k}\|^2\\
	&=\sum_{i=1}^{n}\Big\|\vx_k^{(i)}-\frac{1}{n}\sum_{i=1}^n\vx_k^{(i)}\Big\|^2=\Big\|(I-\frac{1}{n}\bone\bone^\T)X_k\Big\|_F^2\\
	&\leq2\|(I-\bone\pi_A^\T)X_k\|_F^2+2\Big\|(\bone\pi_A^\T-\frac{1}{n}\bone\bone^\T)X_k\Big\|_F^2\\
	&= 2\|X_k-\bone\bar\vx_k^\T\|_F^2+2\Big\|(\frac{1}{n}\bone\bone^\T-\bone\pi_A^\T)(X_k-\bone\bar\vx_k^\T)\Big\|_F^2\\
	&\leq  2n\|X_k-\bone\bar\vx_k^\T\|_F^2
	\ena
	where $X_k$ is defined in \eqref{eq1_sec4}, $\bar\vx_k$ and $\pi_A$ are defined in Theorem \ref{theo1}. The first inequality uses the relation $\|\va+\vb\|^2\leq 2\|\va\|^2+2\|\vb\|^2$, and the last inequality follows from $\|\frac{1}{n}\bone\bone^\T-\bone\pi_A^\T\|_F^2\leq n-1$.

	On the other hand, it follows from \eqref{eq3_sec1} and \eqref{eq2_alg} that 
	\bea 
	&\hspace{-1cm}\sum_{i=1}^{n}\vw_k^{(i)}-\vd=-\sum_{i=1}^{n}\nabla f_i(\vx_{k}^{(i)})\\
	&=-\Big(\nabla f(\bar\vx_{k})+\sum_{i=1}^{n}(\nabla f_i(\vx_{k}^{(i)})-\nabla f_i(\bar \vx_{k}))\Big)
	\ena
	Taking the norm on both sides yields that
	\begin{align}\label{eq2_thm3}
	&\|\sum_{i=1}^{n}\vw_k^{(i)}-\vd\|^2\\
	&\leq2\|\nabla f(\bar\vx_{k})\|^2+2n\sum_{i=1}^{n}\|\nabla f_i(\vx_{k}^{(i)})-\nabla f_i(\bar \vx_{k})\|^2\\
	&\leq2\|\nabla f(\bar\vx_{k})\|^2+\frac{2n}{\mu^2}\sum_{i=1}^{n}\|\vx_{k}^{(i)}-\bar \vx_{k}\|^2\\
	&=2\|\nabla f(\bar\vx_{k})\|^2+\frac{2n}{\mu^2}\|X_k-\bone\bar\vx_k^\T\|_F^2
	\end{align}
	where we use $\|\va+\vb\|^2\leq 2\|\va\|^2+2\|\vb\|^2$ again and the Cauchy-Schwarz inequality to obtain the first inequality, and the second inequality follows from \eqref{lsmooth}. Combining \eqref{eq1_thm3} and \eqref{eq2_thm3} implies that
	\bea
	&\sum_{i=1}^{n}\|\nabla F(\vw_{k}^{(i)})-\overline{\nabla F_k}\|^2+\|\sum_{i=1}^{n}\vw_k^{(i)}-\vd\|^2\\
	&\leq 2\|\nabla f(\bar\vx_{k})\|^2+\frac{2n(1+\mu^2)}{\mu^2}\|X_k-\bone\bar\vx_k^\T\|_F^2
	\ena
	The desired result then follows from Theorem \ref{theo1}.

	The linear convergence rate in the presence of Lipschitz smoothness can be similarly obtained by following the linear convergence of PPG for strongly convex and Lipschitz smooth objective functions (\cite[Theorem 1]{xin2018linear} or \cite[Theorem 1]{pu2018push}), which is omitted to save space.
\end{proof}

Theorem \ref{theo3} shows that the DDGT converges at a sublinear rate $O({1}/{k})$ for strongly convex objective functions, and achieves a linear convergence rate if Lipschitz smoothness is further satisfied. In view of Theorem \ref{theo1}, the explicit form of the term corresponding to $O({1}/{k^2})$ in Theorem \ref{theo3} can be obtained after tedious computations.

\section{Numerical Experiments}\label{sec5}

This section validates our theoretical results and compares the DDGT with existing algorithms via simulation. More precisely, we compare the DDGT with the algorithms in \cite{xu2017distributed,li2018convergence} and \cite[Mirror-Push-DIGing]{nedic2018improved}. Note that \cite{nedic2018improved} does not provide convergence guarantee for Mirror-Push-DIGing, \cite{xu2017distributed} has no convergence rate results, and \cite{li2018convergence} only shows the convergence rate for strongly convex and Lipschitz smooth cost functions. Moreover, the algorithm in \cite{li2018convergence} involves solving a subproblem similar to \eqref{eq2_alg} per iteration and \cite{xu2017distributed} adopts the update in \eqref{eq2c_alg} which is a special case of \eqref{eq2_alg}, and hence the computational complexities of the two algorithms  are similar to DDGT per iteration. In contrast,  Mirror-Push-DIGing \cite{nedic2018improved} requires computing a proximal operator, which may have higher computational costs.

We test these algorithms over $126$ nodes connected via a directed network, which is a real Email network \cite{konect:2016:radoslaw_email,konect}. Each node $i$ is associated with a local quadratic cost function $F_i(w_i)=a_i(w_i-b_i)^2$ where $a_i\sim\cU(0,1)$ and $b_i\sim\cN(0,4)$ are randomly sampled. Note that the quadratic cost function is commonly used in the literature \cite{xu2017distributed,li2018convergence,nedic2018improved}. The global constraint is $\sum_{i=1}^{126}w_i=50$.

We first test the case without local constraints by setting $\cW_i=\bR^m$. The stepsize used for each algorithm is tuned via a grid search\footnote{The grid search scheme works as follows. For each algorithm, we select a ``good" stepsize by inspection, and then gradually increase and decrease stepsizes around the selected one with an equal grid size, respectively. Then, we find the fastest one among all the tried stepsizes. }, and all initial conditions are randomly set. Fig. \ref{fig1} depicts the decay of distance between $\vw_{k}^{(i)}$ and the optimal solution with respect to the number of iterations. It clearly shows that the DDGT has a linear convergence rate and converges faster than algorithms in \cite{xu2017distributed,li2018convergence} and \cite{nedic2018improved}.

To validate the theoretical result for strongly convex cost functions without Lipschitz smoothness, we test the algorithms with a quartic local cost function $F_i(w_i)=a_i(w_i-b_i)^2+c_i(w_i-d_i)^4$, where $c_i\sim\cU(0,10)$ and $d_i\sim\cN(0,4)$ are randomly sampled. Clearly,    this function is strongly convex but not Lipschitz smooth. All other settings remain the same and the result is plotted in Fig. \ref{fig2}, where the Mirror-Push-DIGing \cite{nedic2018improved} is not included because its proximal operator is very time-consuming, and an approximate solution for the proximal operator often leads to a poor performance of the algorithm. The dotted line in Fig. \ref{fig2} is the sequence $\{100/k\}$ with $k$ the number of iterations. We can observe that the convergence rates of all algorithms are slower than that in Fig. \ref{fig1}, but the DDGT still outperforms the other two algorithms. Moreover, it is interesting to observe that the DDGT and the algorithm in \cite{li2018convergence}  have near-linear convergence rate, though the theoretical convergence rate for the DDGT is $O(1/k)$.

Finally, we study the effect of local constraints on the convergence rate. To this end, we assign each node a local constraint $-2\leq w_i\leq 2$, and test all algorithms with the setting of Fig. \ref{fig2}. The result is depicted in Fig. \ref{fig3}, which shows that the convergence of the DDGT is essentially not affected, while the algorithm in \cite{li2018convergence} is heavily slowed compared with that in Fig. \ref{fig2}.

\begin{figure}[!t]
	\centering
	\includegraphics[width=0.6\linewidth]{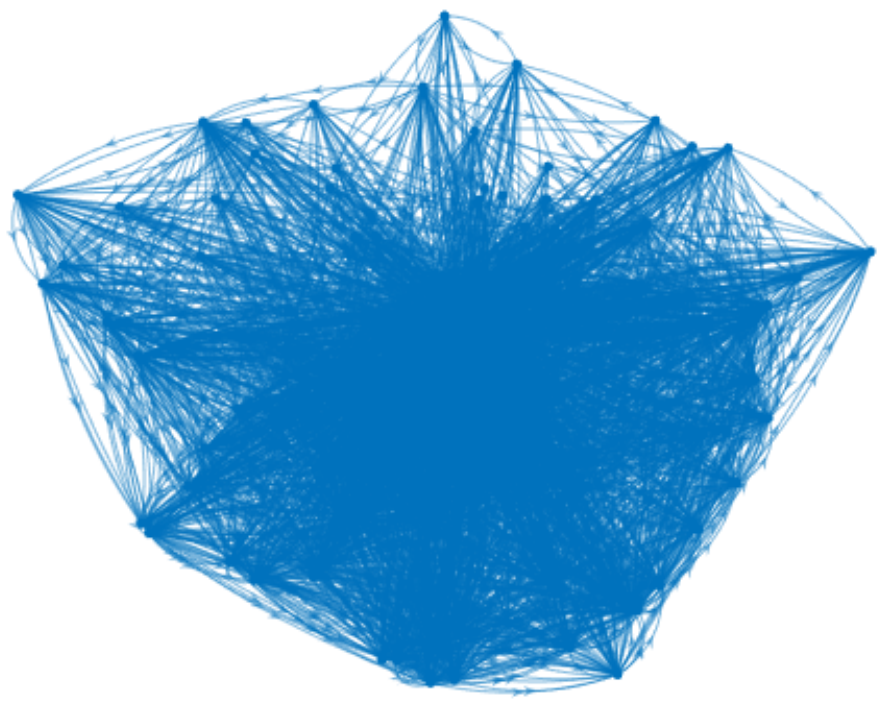}
	\caption{The communication network in \cite{konect:2016:radoslaw_email,konect}.}
	\label{fig4}
\end{figure}

\begin{figure}[!t]
	\centering
	\includegraphics[width=0.8\linewidth]{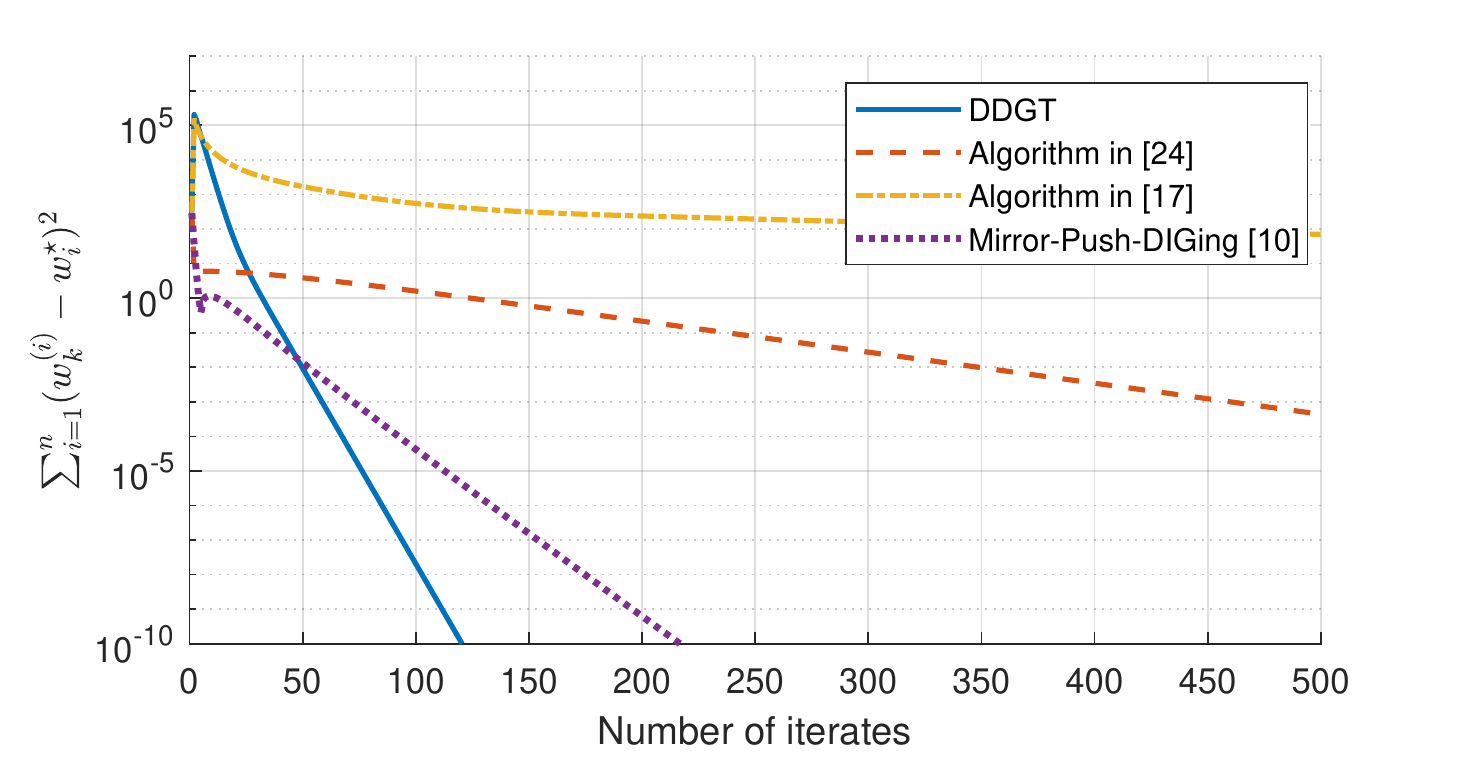}
	\caption{Convergence rate w.r.t the number of iterations of different algorithms with quadratic cost function $F_i(w_i)=a_i(w_i-b_i)^2$.}
	\label{fig1}
\end{figure}

\begin{figure}[!t]
	\centering
	\includegraphics[width=0.8\linewidth]{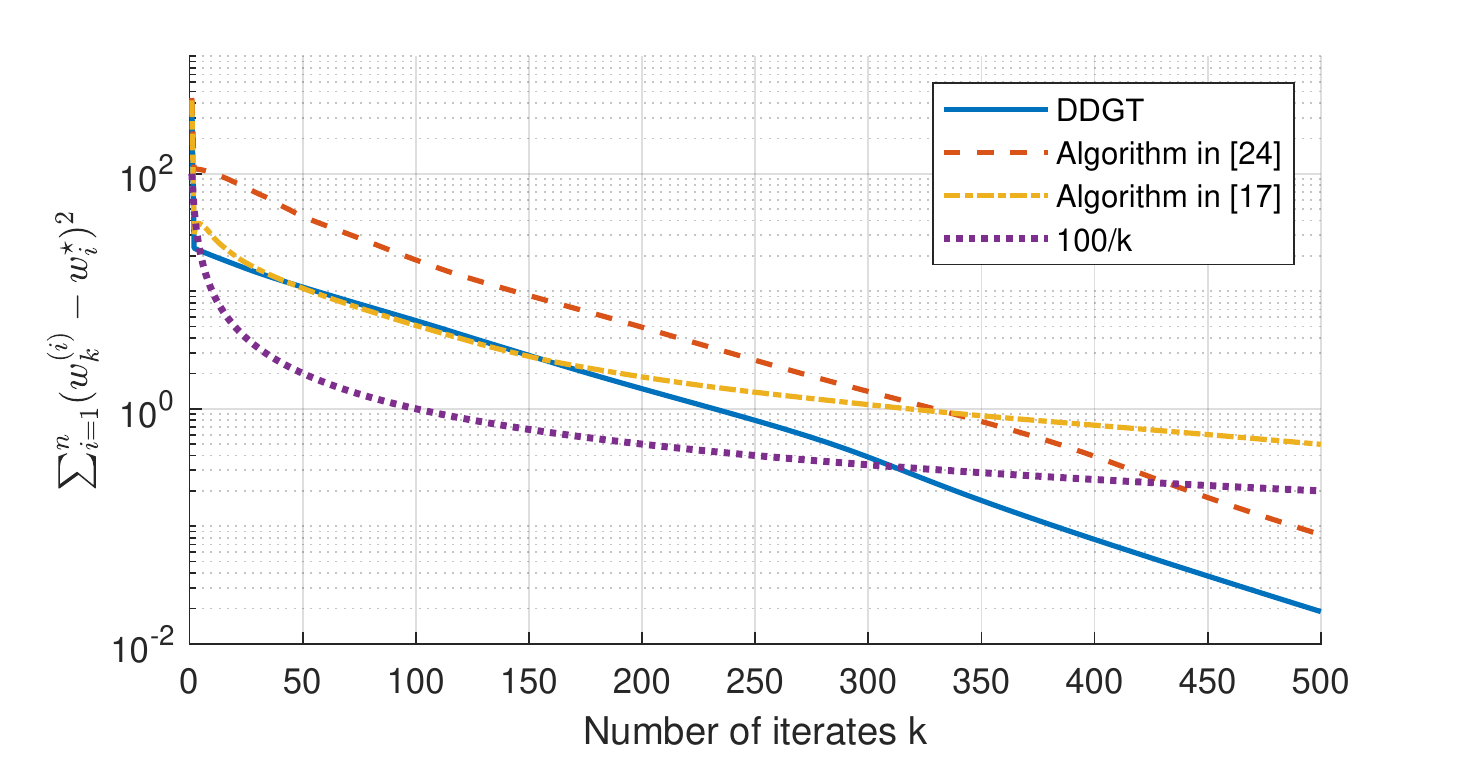}
	\caption{Convergence rate w.r.t the number of iterations of different algorithms with quartic cost function $F_i(w_i)=a_i(w_i-b_i)^2+c_i(w_i-d_i)^4$.}
	\label{fig2}
\end{figure}

\begin{figure}[!t]
	\centering
	\includegraphics[width=0.8\linewidth]{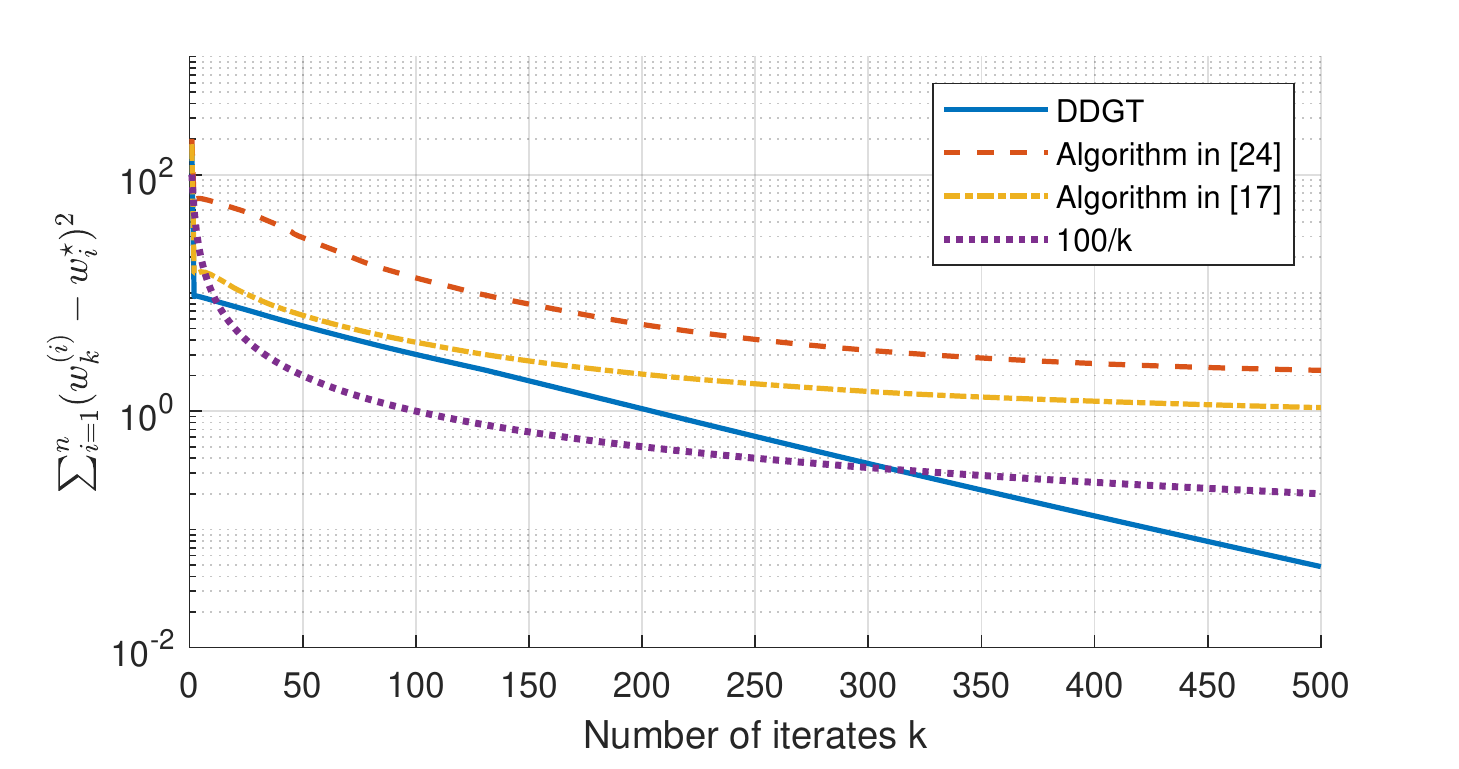}
	\caption{Convergence rate w.r.t the number of iterations of different algorithms with quartic cost function $F_i(w_i)=a_i(w_i-b_i)^2+c_i(w_i-d_i)^4$ and local constraint $-2\leq w_i\leq 2,\forall i$.}
	\label{fig3}
\end{figure}

\section{Conclusion}\label{sec6}

We proposed the DDGT for distributed resource allocation problems (DRAPs) over directed unbalanced networks. Convergence results are provided by exploiting the strong duality of DRAPs and distributed optimization problems, and taking advantage of the PPG algorithm. We studied the convergence and convergence rate of PPG for non-convex problems and obtained that  the DDGT converges linearly for strongly convex and Lipschitz smooth objective functions, and sub-linearly without the Lipschitz smoothness. Future works are to provide tighter bounds for the convergence rate, design asynchronous versions  \cite{zhang2019asyspa,zhang2019asynchronous}, study quantized communication \cite{zhang2018distributed}, and design accelerated algorithms \cite{qu2020accelerated}. In particular, an interesting idea to accelerate the DDGT is to add a vanishing strongly convex regularization term to the dual problems of DRAPs, which may allow a larger stepsize in the early stage and hence possibly lead to faster convergence.

\section*{Acknowledgment}
The authors would like to thank the Associate Editor and anonymous reviewers for their very constructive comments, which greatly improved the quality of this work.

\appendix

\subsection{Preliminary results on stochastic matrices}
We first introduce three lemmas are from \cite{pu2018push,xin2018linear}.
\begin{lemma}[\cite{xin2018linear,pu2018distributed}]\label{lemma1}
	Suppose Assumption \ref{assum2} holds. The matrix $A$ has a unique unit nonnegative left eigenvector $\pi_A$ w.r.t. eigenvalue 1, i.e., $\pi_A^\T A=\pi_A^\T$ and $\pi_A^\T\bone=1$. The matrix $B$ has a unique unit right eigenvector $\pi_B$ w.r.t. eigenvalue 1, i.e., $B\pi_B =\pi_B$ and $\pi_B^\T\bone=1$.
\end{lemma}

The proof of Lemma \ref{lemma1} follows from the Perron-Frobenius theorem and can be found in \cite{pu2018push,xin2018linear}.

\begin{lemma}[\cite{horn2012matrix},\cite{pu2018push,xin2018linear}]\label{lemma2}
	Suppose Assumption \ref{assum2} holds. There exist matrix norms $\|\cdot\|_\sA$ and $\|\cdot\|_\sB$ such that $\sigma_\sA\triangleq \|A-{\bone\pi_A^\T}\|_\sA<1$ and $\sigma_\sB\triangleq \|B-{\pi_B\bone^\T}\|_\sB<1$. Moreover, $\sigma_\sA$ and $\sigma_\sB$ can be arbitrarily close to the second largest absolute value of the eigenvalues of $A$ and $B$, respectively.
\end{lemma}

A method to construct such matrix norms can be found in the proof of  Lemma 5.6.10 in \cite{horn2012matrix}.

\begin{lemma}[\cite{pu2018push,xin2018linear}]\label{lemma3}
	There exist constants $\delta_{\sF\sA},\delta_{\sA\sF},\delta_{\sF\sB}$ and $\delta_{\sB\sF}$ such that for any $X\in\bR^{n\times n}$, we have
	\bea
		&\|X\|_F\leq\delta_{\sF\sA}\|X\|_\sA,\ \|X\|_F\leq\delta_{\sF\sB}\|X\|_\sB\\
		&\|X\|_\sA\leq\delta_{\sA\sF}\|X\|_F,\ \|X\|_\sB\leq\delta_{\sB\sF}\|X\|_F
	\ena
\end{lemma}

Lemma \ref{lemma3} is a direct result of the norm equivalence theorem. If $A$ and $B$ are symmetric, which means the network is undirected, then $\delta_{\sA \sF}=\delta_{\sB \sF}=1$ and $\delta_{\sF\sA}=\delta_{\sF\sB}=\sqrt{n}$.

Note that the norm $\|\cdot\|_\sA$ defined in Lemma \ref{lemma2} is only for matrices in $\bR^{n\times n}$. To facilitate presentation, we slightly abuse the notation and define a \emph{vector} norm $\|\vx\|_\sA\triangleq\|\frac{1}{\sqrt{n}}\vx\bone^\T\|_\sA$ for any $\vx\in\bR^n$, where the norm in the right-hand-side is the matrix norm defined in Lemma \ref{lemma2}. Then, we have
\bee
	\|M\vx\|_\sA=\|\frac{1}{\sqrt{n}}M\vx\bone^\T\|_\sA\leq\|M\|_\sA\Big\|\frac{\vx\bone^\T}{\sqrt{n}}\Big\|_\sA=\|M\|_\sA\|\vx\|_\sA
\ene 
where the first equality is by definition and the inequality follows from the sub-multiplicativity of matrix norms. Moreover, for any matrix $X=[\vx_1,\cdots,\vx_m]\in\bR^{n\times m}$, define the matrix norm $\|X\|_\sA=\sqrt{\sum_{i=1}^{m}\|\vx_i\|_\sA^2}$. Recall that $n\times m$ is the dimension of $X$ and hence the definition is distinguished from that in Lemma \ref{lemma2}. We have
\bea
	\|MX\|_\sA&=\|[M\vx_1,\cdots,M\vx_m]\|_\sA=\sqrt{\sum\nolimits_{i=1}^{m}\|M\vx_i\|_\sA^2}\\
	&\leq\sqrt{\sum\nolimits_{i=1}^{m}\|M\|_\sA^2\|\vx_i\|_\sA^2}=\|M\|_\sA\|X\|_\sA.
\ena
Therefore, for any $M\in\bR^{n\times n}$, $X\in\bR^{n\times m}$, and $\vx\in\bR^n$, the following relation holds
\bee\label{eq1_proof}
	\|MX\|_\sA\leq\|M\|_\sA\|X\|_\sA,\ \|M\vx\|_\sA\leq\|M\|_\sA\|\vx\|_\sA.
\ene
Similarly, we can obtain such a relation based on the matrix norm $\|\cdot\|_\sB$ defined in Lemma \ref{lemma2}.

Next, we define three important auxiliary variables:
\bea\label{eq2_proof}
	\bar\vx_k\triangleq X_k^\T\pi_A,\
	\bar\vy_k\triangleq Y_k^\T\pi_A,\ \hat{\vy}_k\triangleq Y_k^\T\bone\overset{\eqref{ppgb}}=\nabla \vf_k^\T\bone
\ena
where $\bar\vx_k$ is a weighted average of $\vx_k^{(i)}$ that is identical to the one defined in Theorem \ref{theo1}, $\bar\vy_k$ is a weighted average of $\vy_k^{(i)}$, and $\hat\vy_k$ is the sum of $\vy_k^{(i)}$.

Finally, for any $X=[\vx^{(1)},\cdots,\vx^{(n)}]^\T\in\bR^{n\times m}$, let
\bee
	\nabla\vf(X)=[\nabla f_1(\vx^{(1)}),\cdots,\nabla f_n(\vx^{(n)})]^\T\in\bR^{n\times m},
\ene
and let $\rho(X)$ denote the spectral radius of matrix $X$.

\subsection{Proof of Theorem \ref{theo1}}
{\bf Step 1: Bound $\|X_{k}-\bone\bar\vx_{k}^\T\|_\sA$ and $\|Y_{k}-\pi_B\hat\vy_{k}^\T\|_\sB$}

It follows from \eqref{alg} that
\begin{align}\label{eq1_s2}
	&\|X_{k+1}-\bone\bar\vx_{k+1}^\T\|_\sA\\
	&=\|AX_k-\bone\bar\vx_{k}^\T-\alpha(A-\bone\pi_A^\T)Y_k\|_\sA\\
	&=\left\|(A-\bone\pi_A^\T)[(X_k-\bone\bar\vx_{k}^\T)-\alpha (Y_k-\pi_B\hat{\vy}_k^\T)-\alpha \pi_B\hat{\vy}_k^\T]\right\|_\sA\\
	&\leq\sigma_\sA\|X_k-\bone\bar\vx_{k}^\T\|_\sA+\alpha\sigma_\sA\|Y_k-\pi_B\hat{\vy}_k^\T\|_\sA+\alpha\sigma_\sA\|\pi_B\hat{\vy}_k^\T\|_\sA\\
	&\leq\sigma_\sA\|X_k-\bone\bar\vx_{k}^\T\|_\sA+\alpha\sigma_\sA\delta_{\sA\sF}\delta_{\sF\sB}\|Y_k-\pi_B\hat{\vy}_k^\T\|_\sB\\
	&\quad+\alpha\sigma_\sA\delta_{\sA\sF}\|\bone^\T(\nabla \vf(X_k)-\nabla \vf(\bone\bar\vx_k^\T)+\bone^\T\nabla \vf(\bone\bar\vx_k^\T))\|\\
	&\leq\alpha\sigma_\sA\delta_{\sA\sF}\delta_{\sF\sB}\|Y_k-\pi_B\hat{\vy}_k^\T\|_\sB+{\alpha\sigma_\sA\delta_{\sA\sF}}\|\nabla f(\bar\vx_{k})\|\\
	&\quad+(\sigma_\sA+\alpha\sigma_\sA\delta_{\sA\sF}\delta_{\sF\sA}L\sqrt{n})\|X_k-\bone\bar\vx_{k}^\T\|_\sA
\end{align}
where we use Lemma \ref{lemma2} and \eqref{eq1_proof} to obtain the first inequality, the second inequality is from Lemma \ref{lemma3} and \eqref{eq2_proof}, and the last inequality follows from the $L$-Lipschitz smoothness.

Now we bound $\|Y_{k}-\pi_B\hat\vy_{k}^\T\|_\sB$. From \eqref{ppg} we have
\bea\label{eq2_s1}
	&\|Y_{k+1}-\pi_B\hat\vy_{k+1}^\T\|_\sB\\
	&= \|BY_k-\pi_B\hat\vy_{k}^\T+(\nabla \vf_{k+1}-\nabla \vf_{k}) -(\pi_B\hat\vy_{k+1}^\T-\pi_B\hat\vy_{k}^\T)\|_\sB\\
	&=\|(B-\pi_B\bone^\T)(Y_k-\pi_B\hat\vy_{k}^\T)+(I-\pi_B\bone^\T)(\nabla \vf_{k+1}-\nabla \vf_{k})\|_\sB\\
	&\leq\sigma_\sB\|Y_k-\pi_B\hat\vy_{k}^\T\|_\sB+L\delta_{\sB\sF}\|I-\pi_B\bone^\T\|_\sB\|X_{k+1}-X_k\|_F\\
	&\leq\sigma_\sB\|Y_k-\pi_B\hat\vy_{k}^\T\|_\sB+L\delta_{\sB\sF}\|X_{k+1}-X_k\|_F.
\ena
where the last inequality follows from $\|I-\pi_B\bone^\T\|_\sB=1$, which can be readily obtained from the construction of the norm $\|\cdot\|_\sB$ \cite[Lemma 5.6.10]{horn2012matrix}. Moreover, it follows from \eqref{ppga} that
\bea
	&\|X_{k+1}-X_k\|_F=\|AX_{k}-X_k-\alpha AY_k\|_F \\ &=\|(A-I)(X_k-\bone\bar\vx_k^\T)-\alpha AY_k\|_F \\
	&\leq\left\|A-I\right\|\|X_k-\bone\bar\vx_k^\T\|_F +\alpha \|A(Y_k-\pi_B\hat\vy_{k}^\T+\pi_B\hat\vy_{k}^\T)\|_F \\
	&\leq 2\sqrt{n}\|X_k-\bone\bar\vx_k^\T\|_F +\alpha \|A\|(\|Y_k-\pi_B\hat\vy_{k}^\T\|_F +\|\pi_B\hat\vy_{k}^\T\|_F) \\
	&\leq 2\sqrt{n}\delta_{\sF\sA}\|X_k-\bone\bar\vx_k^\T\|_\sA +\alpha \sqrt{n}(\delta_{\sF\sB}\|Y_k-\pi_B\hat\vy_{k}^\T\|_\sB +\|\hat\vy_{k}^\T\|)\\
	&\leq 2\sqrt{n}\delta_{\sF\sA}\|X_k-\bone\bar\vx_k^\T\|_\sA +\alpha \sqrt{n}\delta_{\sF\sB}\|Y_k-\pi_B\hat\vy_{k}\|_\sB\\
	&\quad+\alpha \sqrt{n}\|\bone^\T(\nabla \vf(X_k)-\nabla \vf(\bone\bar\vx_k^\T)+\bone^\T\nabla \vf(\bone\bar\vx_k^\T))\|\\
	&\leq (\alpha\L n\delta_{\sF\sA}+2\sqrt{n}\delta_{\sF\sA})\|X_k-\bone\bar\vx_k^\T\|_\sA\\
	&\quad+\alpha \sqrt{n}\delta_{\sF\sB}\|Y_k-\pi_B\hat\vy_{k}\|_\sB+\alpha\sqrt{n}\|\nabla f(\bar\vx_{k})\|
\ena
where we used $\|A\|\leq\sqrt{n}$. The above relation combined with \eqref{eq2_s1} yields
\bea\label{eq1_s1}
	&\|Y_{k+1}-\pi_B\hat\vy_{k+1}^\T\|_\sB\\
	&\leq (\sigma_\sB+\L \alpha\sqrt{n}\delta_{\sB\sF}\delta_{\sF\sB}) \|Y_k-\pi_B\hat\vy_{k}^\T\|_\sB  \\
	&\quad+\sqrt{n}L\delta_{\sB\sF}\delta_{\sF\sA}(2+\sqrt{n}\L\alpha) \|X_k-\bone\bar\vx_k^\T\|_\sA\\
	&\quad+\alpha\sqrt{n}L\delta_{\sB\sF}\|\nabla f(\bar\vx_{k})\|.
\ena

Combing \eqref{eq1_s2} and \eqref{eq1_s1} implies the following linear matrix inequality
\bea\label{eq_lmi}
	&\underbrace{\begin{bmatrix}
			\|X_{k+1}-\bone\bar\vx_{k+1}^\T\|_\sA \\
			\|Y_{k+1}-\pi_B\hat\vy_{k+1}^\T\|_\sB
		\end{bmatrix}}_{\textstyle \triangleq  \vz_{k+1}}\preccurlyeq
	\underbrace{\begin{bmatrix}
			P_{11} & P_{12} \\
			P_{21} & P_{22}
		\end{bmatrix}}_{\textstyle \triangleq  P}
	\underbrace{\begin{bmatrix}
			\|X_{k}-\bone\bar\vx_{k}^\T\|_\sA \\
			\|Y_{k}-\pi_B\hat\vy_{k}^\T\|_\sB
		\end{bmatrix}}_{\textstyle \triangleq \vz_{k}}\\
	&\quad+\underbrace{\begin{bmatrix}
			\alpha\sigma_\sA\delta_{\sA\sF}\|\nabla f(\bar\vx_{k})\| \\
			\alpha\sqrt{n}\L \delta_{\sB\sF}\|\nabla f(\bar\vx_{k})\|
		\end{bmatrix}}_{\textstyle \triangleq \vu_{k}}\\
\ena
where $\preccurlyeq$ denotes the element-wise less than or equal sign and
\bea
	&P_{11}=\sigma_\sA+\alpha\sigma_\sA\delta_{\sA\sF}\delta_{\sF\sA}L\sqrt{n},&&P_{12}=\alpha\sigma_\sA\delta_{\sA\sF}\delta_{\sF\sB}\\
	&P_{21}=\sqrt{n}L\delta_{\sB\sF}\delta_{\sF\sA}(2+\sqrt{n}\L\alpha),&&P_{22}=\sigma_\sB+\L \alpha\sqrt{n}\delta_{\sB\sF}\delta_{\sF\sB}
\ena
Note that $\rho(P)<1$ for sufficiently small $\alpha$, since
\bee
	\lim_{\alpha\ra0}P=
	\begin{bmatrix}
		\sigma_\sA                                & 0          \\
		2L\sqrt{n}\delta_{\sB \sF}\delta_{\sF\sA} & \sigma_\sB
	\end{bmatrix}
\ene
has spectral radius smaller than 1.

The linear matrix inequality \eqref{eq_lmi} implies that
\bee\label{eq1_s3}
	\vz_{k}\preccurlyeq P^{k-1}\vz_1+\sum_{t=1}^{k-1}P^{t-1}\vu_{k-t}.
\ene
Let $\theta_1$ and $\theta_2$ be the two eigenvalues of $P$ such that $|\theta_2|>|\theta_1|$, and $\theta\triangleq\rho(P)=|\theta_2|$, then $P$ can be diagonalized as
\bee\label{eq3_s1}
	P=T\Lambda T^{-1},\
	\Lambda=\begin{bmatrix}
		\theta_1 & 0        \\
		0        & \theta_2
	\end{bmatrix}.
\ene
Let
$
	\Psi =\sqrt{(P_{11}-P_{22})^2+4P_{12}P_{21}}. 
$
Note that the analysis so far holds if $\sigma_\sA$ is replaced by any value in $(\sigma_\sA,1)$ (similar for $\sigma_\sB$), and hence we assume without loss of generality that $\sigma_\sA\neq\sigma_\sB$ to simplify presentation. In that case, $\Psi$ is lower bounded by some positive value that is independent of $\alpha$, say $\underline{\Psi}$. With some tedious calculations, we have
\bea\label{eq4_s3}
	\centering
	\theta_1&=\frac{P_{11}+P_{22}-\Psi}{2}\\
	\theta&=\theta_2=\frac{P_{11}+P_{22}+\Psi}{2}\\
	&=\frac{1}{2}(\sigma_\sA+\sigma_\sB+L\alpha\sqrt{n}(\delta_{\sB\sF}\delta_{\sF\sB}+\sigma_\sA\delta_{\sA\sF}\delta_{\sF\sA})+\Psi).
\ena
To let $\theta=\theta_2<1$, it is sufficient for $\alpha$ to satisfy
\bee\label{eq_alpha1}
	\alpha<\frac{(1-\sigma_\sA)(1-\sigma_\sB)}{2(\sqrt{n}\L \sigma_\sA\delta_{\sA\sF}\delta_{\sF\sA}+1)(\sqrt{n}\L\delta_{\sB\sF}\delta_{\sF\sB}+1)}.
\ene
Moreover, $T$ and $T^{-1}$ in \eqref{eq3_s1} can be expressed in an explicit form
\bea
	\centering
	&T=\begin{bmatrix}
		\frac{P_{11}-P_{22}-\Psi}{2P_{21}} & \frac{P_{11}-P_{22}+\Psi}{2P_{21}} \\
		1                                  & 1
	\end{bmatrix}, T^{-1}=\begin{bmatrix}
		-\frac{P_{21}}{\Psi} & \frac{P_{11}-P_{22}+\Psi}{2\Psi} \\
		\frac{P_{21}}{\Psi}  & \frac{P_{22}-P_{11}+\Psi}{2\Psi}
	\end{bmatrix}
\ena
It then follows from \eqref{eq3_s1} that
\bea\label{eq4_s1}
	&0\curlyeqprec P^k=T\Lambda^k T^{-1}\\
	&=\begin{bmatrix}
		\frac{\theta_1^k+\theta_2^k}{2}+\frac{(P_{11}-P_{22})(\theta_2^k-\theta_1^k)}{2\Psi} & \frac{P_{12}}{\Psi}(\theta_2^k-\theta_1^k)                                           \\
		\frac{P_{21}}{\Psi}(\theta_2^k-\theta_1^k)                                           & \frac{\theta_1^k+\theta_2^k}{2}+\frac{(P_{11}-P_{22})(\theta_1^k-\theta_2^k)}{2\Psi} \\
	\end{bmatrix}\\
	&\curlyeqprec
	\theta^k\begin{bmatrix}
		1 & {(n\L^2\underline{\Psi})}^{-1}                        \\
		3\sqrt{n}L\delta_{\sB\sF}\delta_{\sF\sA}/\underline{\Psi}                           & 1 \\
	\end{bmatrix}
\ena
where we used $|P_{11}-P_{22}|\leq \Psi,\Psi\geq \underline{\Psi}$, and the bound \eqref{eq_alpha1} to obtain the inequality.

Combining \eqref{eq_lmi}, \eqref{eq1_s3} and \eqref{eq4_s1} yields that
\bea\label{eq3_s3}
	\|X_{k}-\bone\bar\vx_{k}^\T\|_F &\leq c_0\theta^{k-1}+c_1\alpha\sum_{t=1}^{k-1}\theta^{t-1} \|\nabla f(\bar\vx_{k-t})\|\\
	\|Y_{k}-\pi_B\hat\vy_{k}^\T\|_F &\leq c_2\theta^{k-1}+c_3\alpha\sum_{t=1}^{k-1}\theta^{t-1} \|\nabla f(\bar\vx_{k-t})\|\\
\ena
where $c_0,c_1,c_2$ and $c_3$ are constants given as follows
\bea\label{eq6_s3}
	c_0&=\frac{\|Y_{1}-\pi_B\hat\vy_{1}^\T\|_\sB}{nL^2\underline \Psi}\leq\frac{\delta_{\sB\sF}}{nL^2\underline{\Psi}}\Big(\sum_{i=1}^n\|\nabla f_i(\vx_1)\|^2\Big)^{1/2}\\
	c_1&={\sigma_\sA\delta_{\sA \sF}}+\frac{\delta_{\sB\sF}\delta_{\sF\sB}}{nL\underline{\Psi}}\\
	c_2&=\|Y_{1}-\pi_B\hat\vy_{1}^\T\|_\sB\leq\delta_{\sB\sF}\Big(\sum_{i=1}^n\|\nabla f_i(\vx_1)\|^2\Big)^{1/2}\\
	c_3&=\frac{3\sqrt{n}\L\sigma_\sA\delta_{\sB\sF}\delta_{\sF\sA}\delta_{\sA\sF}}{\underline\Psi}+{\sqrt{n}L\delta_{\sB\sF}}.
\ena

{\bf Step 2: Bound $ \|\bar{\vy}_k\|^2 $}

From \eqref{ppg} and the $L$-Lipschitz smoothness, we have
\bea\label{eq2_s2}
	f(\bar\vx_{k+1})\leq f(\bar\vx_{k})-\alpha\nabla f(\bar\vx_{k})^\T\bar\vy_{k}+\frac{L\alpha^2}{2}\|\bar\vy_{k}\|^2.
\ena
Note that
\bea\label{eq4_s2}
	\bar\vy_k&=Y_k^\T\pi_A=(Y_k-\pi_B\hat\vy_{k}^\T+\pi_B\hat\vy_{k}^\T)^\T\pi_A\\
	&=(Y_k-\pi_B\hat\vy_{k}^\T)^\T\pi_A+ Y_k^\T\bone\pi_B^\T\pi_A\\
	&=(Y_k-\pi_B\hat\vy_{k}^\T)^\T\pi_A+(\nabla\vf(X_k)-\nabla \vf(\bone\bar\vx_{k}^\T))^\T\bone\pi_B^\T\pi_A\\
	&\quad+{\pi_B^\T\pi_A}\nabla f(\bar\vx_{k})
\ena
where we used the relation $Y_k^\T\bone=\nabla\vf(X_k)^\T\bone$ and $\nabla\vf(\bone\bar\vx_{k}^\T)^\T\bone=\nabla f(\bar\vx_{k})$. Then, we have
\bea\label{eq5_s2}
	-&\nabla f(\bar\vx_{k})^\T\bar\vy_{k}\\
	&=-\nabla f(\bar\vx_{k})^\T(\nabla\vf(X_k)-\nabla \vf(\bone\bar\vx_{k}^\T))^\T\bone\pi_B^\T\pi_A\\
	&\quad-\nabla f(\bar\vx_{k})^\T(Y_k-\pi_B\hat\vy_{k}^\T)^\T\pi_A-{\pi_B^\T\pi_A}\|\nabla f(\bar\vx_{k})\|^2\\
	&\leq-{\pi_B^\T\pi_A}\|\nabla f(\bar\vx_{k})\|^2+{L}{\sqrt{n}}\|\nabla f(\bar\vx_{k})\|\|X_{k}-\bone\bar\vx_{k}^\T\|_F\\
	&\quad+\|\nabla f(\bar\vx_{k})\|\|Y_{k}-\pi_B\hat\vy_{k}^\T\|_F
\ena
where we used $\|\pi_A\|\leq1$, and the Lipschitz smoothness $\|\nabla\vf(X_k)-\nabla \vf(\bone\bar\vx_{k}^\T)\|_F\leq L\|X_{k}-\bone\bar\vx_{k}^\T\|_F$ to obtain the last inequality.

Moreover, it follows from \eqref{eq4_s2} and the relation $\|\va+\vb+\vc\|^2\leq 3\|\va\|^2+3\|\vb\|^2+3\|\vc\|^2$ that
\bea\label{eq6_s2}
	\|\bar\vy_{k}\|^2&\leq 3\|(Y_k-\pi_B\hat\vy_{k}^\T)^\T\pi_A\|^2+3\|{\pi_B^\T\pi_A}\nabla f(\bar\vx_{k})\|^2\\
	&\quad+3\|(\nabla\vf(X_k)-\nabla \vf(\bone\bar\vx_{k}^\T))^\T\bone\pi_B^\T\pi_A\|^2\\
	&\leq3\|Y_k-\pi_B\hat\vy_{k}^\T\|^2+{3(\pi_B^\T\pi_A)^2}\|\nabla f(\bar\vx_{k})\|^2\\
	&\quad+{3L^2}{n}\|X_{k}-\bone\bar\vx_{k}^\T\|^2.
\ena

{\bf Step 3: Bound $\sum_{t=1}^{k}\|\nabla f(\bar\vx_{t})\|\|X_{t}-\bone\bar\vx_{t}^\T\|_F$ and $\sum_{t=1}^{k} \|X_t-\bone\bar\vx_t^\T\|_F^2 $}

We first bound the summation of the terms  $\|\nabla f(\bar\vx_{t})\|\|X_{t}-\bone\bar\vx_{t}^\T\|_F$ and $\|\nabla f(\bar\vx_{t})\|\|Y_{t}-\pi_B\hat\vy_{t}^\T\|_F$ in \eqref{eq5_s2} over $t=1,\cdots,k$. It follows from \eqref{eq3_s3} that
\bea\label{eq2_s3}
	&\|\nabla f(\bar\vx_{k})\|\|X_{k}-\bone\bar\vx_{k}^\T\|_F\\
	&\leq c_0\theta^{k-1}\|\nabla f(\bar\vx_{k})\|+c_1\alpha\|\nabla f(\bar\vx_{k})\|\sum_{t=1}^{k-1}\theta^{t-1} \|\nabla f(\bar\vx_{k-t})\|\\
\ena

Then, define
\begin{align}
	\vartheta_t&=[\theta^{t-2},\theta^{t-3},\cdots,\theta,1,0,\cdots,0]^\T\in\bR^{k}\\
	\tilde\vartheta_t&=[\underbrace{0,\cdots,0}_{t-1},1,0,\cdots,0]^\T\in\bR^{k}\\
	\upsilon_k&=[\|\nabla f(\bar\vx_{1})\|,\cdots,\|\nabla f(\bar\vx_{k})\|]^\T\in\bR^{k}\\
	\tilde\Theta_k&=\sum_{t=1}^k\vartheta_t\tilde\vartheta_t^\T=\begin{bmatrix}
		0 & 1 & \theta & \cdots & \theta^{k-2} \\
		  & 0 & 1      & \cdots & \theta^{k-1} \\
		  &   & \ddots & \ddots & \vdots       \\
		  &   &        & 0      & 1            \\
		  &   &        &        & 0
	\end{bmatrix}
\end{align}
where $\theta$ is defined in \eqref{eq4_s3}. Note that $\|\nabla f(\bar\vx_t)\|=\upsilon_k^\T\tilde\vartheta_t$ and $\sum_{l=1}^{t-1}\theta^{l-1} \|\nabla f(\bar\vx_{t-l})\|=\upsilon_k^\T\vartheta_t,\forall t\leq k$, which combined with the relation $\|\nabla f(\bar\vx_{k})\|\leq1+\|\nabla f(\bar\vx_{k})\|^2$ and \eqref{eq2_s3} yields
\bea\label{eq7_s2}
	&\sum_{t=1}^{k}\|\nabla f(\bar\vx_{t})\|\|X_{t}-\bone\bar\vx_{t}^\T\|_F\\
	&\leq c_0\sum_{t=1}^{k}\theta^{t-1}(1+\|\nabla f(\bar\vx_{t})\|^2)+c_1\alpha\sum_{t=1}^{k}\|\nabla f(\bar\vx_{t})\|\vartheta_t^\T\upsilon_k\\
	&\leq\frac{c_0}{1-\theta}+{c_0}\sum_{t=1}^{k}\theta^{t-1}\|\nabla f(\bar\vx_{t})\|^2+c_1\alpha\sum_{t=1}^{k}\upsilon_k^\T\tilde\vartheta_t\vartheta_t^\T\upsilon_k\\
	&\leq\frac{c_0}{1-\theta}+{c_0}\sum_{t=1}^{k}\theta^{t-1}\|\nabla f(\bar\vx_{t})\|^2+c_1\alpha\upsilon_k^\T\tilde\Theta_k\upsilon_k.
\ena
The last term $\upsilon_k^\T\tilde\Theta_k\upsilon_k$ in \eqref{eq7_s2} can be bounded by
\bee
	\upsilon_k^\T\tilde\Theta_k\upsilon_k=\upsilon_k^\T\frac{\tilde\Theta_k+\tilde\Theta_k^\T}{2}\upsilon_k\leq\frac{1}{2}\rho(\tilde\Theta_k+\tilde\Theta_k^\T)\|\upsilon_k\|^2\leq\frac{\|\upsilon_k\|^2}{1-\theta}
\ene
where the last inequality follows from 
$
	\rho(\tilde\Theta_k+\tilde\Theta_k^\T)\leq\|\tilde\Theta_k+\tilde\Theta_k^\T\|_1\leq \|\tilde\Theta_k\|_1+\|\tilde\Theta_k\|_\infty\leq\frac{2}{1-\theta}.
$
Thus, we have from \eqref{eq7_s2} that
\bea\label{eq8_s2}
	&\sum_{t=1}^{k}\|\nabla f(\bar\vx_{t})\|\|X_{t}-\bone\bar\vx_{t}^\T\|_F\\
	&\leq\frac{c_0}{1-\theta}+{c_0}\sum_{t=1}^{k}\theta^{t-1}\|\nabla f(\bar\vx_{t})\|^2+\frac{c_1\alpha}{1-\theta}\sum_{t=1}^{k}\|\nabla f(\bar\vx_{t})\|^2
\ena

Similarly, we can bound $\sum_{t=1}^{k}\|\nabla f(\bar\vx_{t})\|\|Y_{t}-\pi_B\hat\vy_{t}^\T\|_F$ as follows,
\bea
	&\sum_{t=1}^{k}\|\nabla f(\bar\vx_{t})\|\|Y_{t}-\pi_B\hat\vy_{t}^\T\|_F\\
	&\leq\frac{c_2}{1-\theta}+{c_2}\sum_{t=1}^{k}\theta^{t-1}\|\nabla f(\bar\vx_{t})\|^2+\frac{c_3\alpha}{1-\theta}\sum_{t=1}^{k}\|\nabla f(\bar\vx_{t})\|^2.
\ena

Next, we bound $\sum_{t=1}^{k} \|X_t-\bone\bar\vx_t^\T\|_F^2 $ and $\sum_{t=1}^{k} \|Y_t-\pi_B\hat\vy_t\|_F^2$. We first consider $\sum_{t=1}^{k} \|X_t-\bone\bar\vx_t^\T\|_F^2 $. For any $k\in\bN$, define
\bea
	\nu_k&=[c_0,c_1\alpha \|\nabla f(\bar\vx_{1})\|,\cdots,c_1\alpha \|\nabla f(\bar\vx_{k-1})\| ]^\T\in\bR^{k}\\
	\phi_t&=[\theta^{t-1},\theta^{t-2},\cdots,\theta,1,0,\cdots,0]^\T\in\bR^{k}\\
	\Theta_k&=\sum_{t=1}^k\phi_t\phi_t^\T\in\bR^{k\times k}
\ena
where the elements are defined in \eqref{eq4_s3} and \eqref{eq6_s3}. Clearly, $\Theta_k$ is nonnegative and positive semi-definite. We have from \eqref{eq3_s3} that $\|X_t-\bone\bar\vx_t^\T\|_F \leq\nu_k^\T\phi_t$, and hence
\bea\label{eq4_s4}
	\sum_{t=1}^{k} \|X_t-\bone\bar\vx_t^\T\|_F^2 \leq \nu_k^\T\Theta_k\nu_k\leq \|\Theta_k\|\|\nu_k\|^2.
\ena
To bound $\|\Theta_k\|$, let $[\Theta_k]_{ij}$ be the element in the $i$-th row and $j$-th column of $\Theta_k$. For any $0<i\leq j\leq k$, we have
\bea
	&[\Theta_k]_{ij}\\
	&=\sum_{t=j-1}^{k}\theta^{t-i+1}\theta^{t-j+1}=\sum_{t=j-1}^{k}\theta^{2t-i-j+2}\\
	&=\frac{\theta^{2j-2}(1-\theta^{2(k-j+2)})}{1-\theta^2}\theta^{2-i-j}=\frac{\theta^{j-i}(1-\theta^{2(k-j+2)})}{1-\theta^2}.
\ena
Since $\Theta_k$ is symmetric, it holds that
\bea
	&\sum_{i=1}^k[\Theta_k]_{ij}=\sum_{i=1}^j[\Theta_k]_{ij}+\sum_{i=j+1}^k[\Theta_k]_{ij}\\
	&=\sum_{i=1}^j[\Theta_k]_{ij}+\sum_{i=j+1}^k[\Theta_k]_{ji}\\
	&=\sum_{i=1}^j\frac{\theta^{j-i}(1-\theta^{2(k-j+2)})}{1-\theta^2}+\sum_{i=j+1}^k\frac{\theta^{i-j}(1-\theta^{2(k-i+2)})}{1-\theta^2}\\
	&\leq\sum_{i=1}^j\frac{\theta^{j-i}}{1-\theta^2}+\sum_{i=j+1}^k\frac{\theta^{i-j}}{1-\theta^2}\\
	&\leq\frac{1}{(1-\theta)(1-\theta^2)}+\frac{\theta}{(1-\theta)(1-\theta^2)}\leq\frac{1}{(1-\theta)^2}
\ena
and we have from the Gershgorin circle theorem that
\bee\label{eq5_s5}
	\|\Theta_k\|\leq\max_j\sum_{i=1}^k[\Theta_k]_{ij}\leq\frac{1}{(1-\theta)^2}.
\ene
It then follows from \eqref{eq4_s4} that
\bea\label{eq7_s4}
&\sum_{t=1}^{k} \|X_t-\bone\bar\vx_t^\T\|_F^2\leq\frac{2}{(1-\theta)^2}\Big[c_0^2+c_1^2\alpha^2\sum_{t=1}^{k-1} \|\nabla f(\bar\vx_{t})\|^2\Big]\\
&\sum_{t=1}^{k} \|Y_t-\pi_B\hat\vy_t\|_F^2\leq\frac{2}{(1-\theta)^2}\Big[c_2^2+c_3^2\alpha^2\sum_{t=1}^{k-1} \|\nabla f(\bar\vx_{t})\|^2\Big].
\ena

{\bf Step 4: Bound $\sum_{t=1}^{k} \|\nabla f(\bar\vx_{t})\|^2$}

Combining \eqref{eq2_s2}, \eqref{eq5_s2} and \eqref{eq6_s2} implies that
\bea\label{eq3_s2}
	&f(\bar\vx_{k+1})\\
	&\leq f(\bar\vx_{k})-\alpha\nabla f(\bar\vx_{k})^\T\bar\vy_{k}+\frac{L\alpha^2}{2}\|\bar\vy_{k}\|^2\\
	&\leq f(\bar\vx_{k})-{\alpha\pi_B^\T\pi_A}\left(1-\frac{3L\alpha\pi_B^\T\pi_A}{2}\right)\|\nabla f(\bar\vx_{k})\|^2\\
	&\quad+\alpha\|\nabla f(\bar\vx_{k})\|\|Y_{k}-\pi_B\hat\vy_{k}^\T\|_F+\frac{3L\alpha^2}{2}\|Y_k-\pi_B\hat\vy_{k}^\T\|_F^2\\
	&\quad+\frac{3L^3\alpha^2}{2}\|X_{k}-\bone\bar\vx_{k}^\T\|_F^2+L\alpha\sqrt{n}\|\nabla f(\bar\vx_{k})\|\|X_{k}-\bone\bar\vx_{k}^\T\|_F.
\ena

Summing  both sides of \eqref{eq3_s2} over $1,\cdots,k$, we have
\bea\label{eq2_s4}
	&{\alpha\pi_B^\T\pi_A}\left(1-\frac{3L\alpha\pi_B^\T\pi_A}{2}\right)\sum_{t=1}^{k} \|\nabla f(\bar\vx_{t})\|^2\\
	&\leq f(\vx_0)-f(\bar\vx_k)+\alpha\sum_{t=1}^{k}\|\nabla f(\bar\vx_{t})\|\|Y_{t}-\pi_B\hat\vy_{t}^\T\|_F\\
	&\quad+\frac{3L\alpha^2}{2}\sum_{t=1}^{k} \Big(\|Y_t-\pi_B\hat\vy_t\|_F^2+L^2\|X_t-\bone\bar\vx_t^\T\|_F^2\Big)\\
	&\quad+\sum_{t=1}^{k}{L\alpha}{\sqrt{n}}\|\nabla f(\bar\vx_{t})\|\|X_{t}-\bone\bar\vx_{t}^\T\|_F
\ena
\bea
	&\leq f(\vx_0)-f^\star+\frac{3L\alpha^2(L^2c_0^2+ c_2^2)}{(1-\theta)^2}+\frac{\alpha(\sqrt{n}Lc_0+c_2)}{(1-\theta)^2}\\
	&\quad+\frac{3L\alpha^4(L^2c_1^2+ c_3^2)}{(1-\theta)^2}\sum_{t=1}^{k} \|\nabla f(\bar\vx_{t})\|^2\\
	&\quad+\frac{\alpha^2(\sqrt{n}Lc_1+c_3)}{(1-\theta)^2}\sum_{t=1}^{k} \|\nabla f(\bar\vx_{t})\|^2\\
	&\quad+{\alpha(\sqrt{n}Lc_0+c_2)}\sum_{t=1}^{k}\theta^{t-1}\|\nabla f(\bar\vx_{t})\|^2
\ena
where the last inequality follows from \eqref{eq8_s2} and \eqref{eq7_s4}.

We can move the terms related to $\sum_{t=1}^{k} \|\nabla f(\bar\vx_{t})\|^2$ in the right-hand-side of \eqref{eq2_s4} to the left-hand-side to bound $\sum_{t=1}^{k} \|\nabla f(\bar\vx_{t})\|^2$. To this end, the stepsize $\alpha$ should satisfy
\bea\label{eq_alpha}
	&\alpha<\\
	&\left(\frac{3L\pi_B^\T\pi_A}{2}+\frac{3L^3c_1^2+3L c_3^2+L\sqrt{n}(c_0+c_1)+c_2+c_3}{(1-\theta)^2\pi_B^\T\pi_A}\right)^{-1}
\ena
which is followed by
\bea\label{eq_gamma}
	\gamma &\triangleq {\alpha\pi_B^\T\pi_A}\left(1-\frac{3L\alpha\pi_B^\T\pi_A}{2}\right.\\
	&\quad\left.-\frac{\alpha(3L^3c_1^2+3L c_3^2+L\sqrt{n}(c_0+c_1)+c_2+c_3)}{(1-\theta)^2\pi_B^\T\pi_A}\right)>0.
\ena
If $\theta^k\leq \frac{\alpha}{1-\theta}$, i.e.,
\bee\label{eq_k0}
	k\geq k_0\triangleq \frac{\ln(\alpha)-\ln(1-\theta)}{\ln(\theta)},
\ene
then it follows from \eqref{eq2_s4} that
\bea
	&\gamma\sum_{t=1}^{k} \|\nabla f(\bar\vx_{t})\|^2\leq f(\vx_0)-f^\star+\frac{3L\alpha^2(L^2c_0^2+ c_2^2)}{(1-\theta)^2}\\
	&\quad+\frac{\alpha(\sqrt{n}Lc_0+c_2)}{(1-\theta)^2}+{\alpha(\sqrt{n}Lc_0+c_2)}\sum_{t=1}^{k_0}\|\nabla f(\bar\vx_{t})\|^2
\ena
Thus, we have
\bea
	&\frac{1}{k}\sum_{t=1}^{k} \|\nabla f(\bar\vx_{t})\|^2\leq \frac{f(\vx_0)-f^\star}{\gamma k}+\frac{3L\alpha^2(L^2c_0^2+ c_2^2)}{\gamma(1-\theta)^2k}\\
	&\quad+\frac{\alpha(\sqrt{n}Lc_0+c_2)(1+\sum_{t=1}^{k_0}\|\nabla f(\bar\vx_{t})\|^2)}{\gamma(1-\theta)^2k}
\ena
which is \eqref{eq1_theo1} in Theorem \ref{theo1}. The inequality \eqref{eq2_theo1}	follows from \eqref{eq7_s4} immediately.

Now we look back at \eqref{eq3_s2}.  Jointly with  \eqref{eq1_theo1}, \eqref{eq8_s2}, \eqref{eq7_s4} and \eqref{eq3_s2}, it follows from the supermartingale convergence theorem \cite[Proposition A.4.4]{bertsekas2015convex} that $f(\bar\vx_k)$ converges. If $f$ is further convex, it follows from the convergence of $\sum_{t=1}^{k}\|\nabla f(\bar\vx_t)\|^2$ that $f(\bar\vx_k)$ converges to $f^\star$.

\bibliographystyle{IEEEtran}
\bibliography{mybibf}         

\begin{IEEEbiography}
	[{\includegraphics[width=1in,height=1.25in,clip,keepaspectratio]{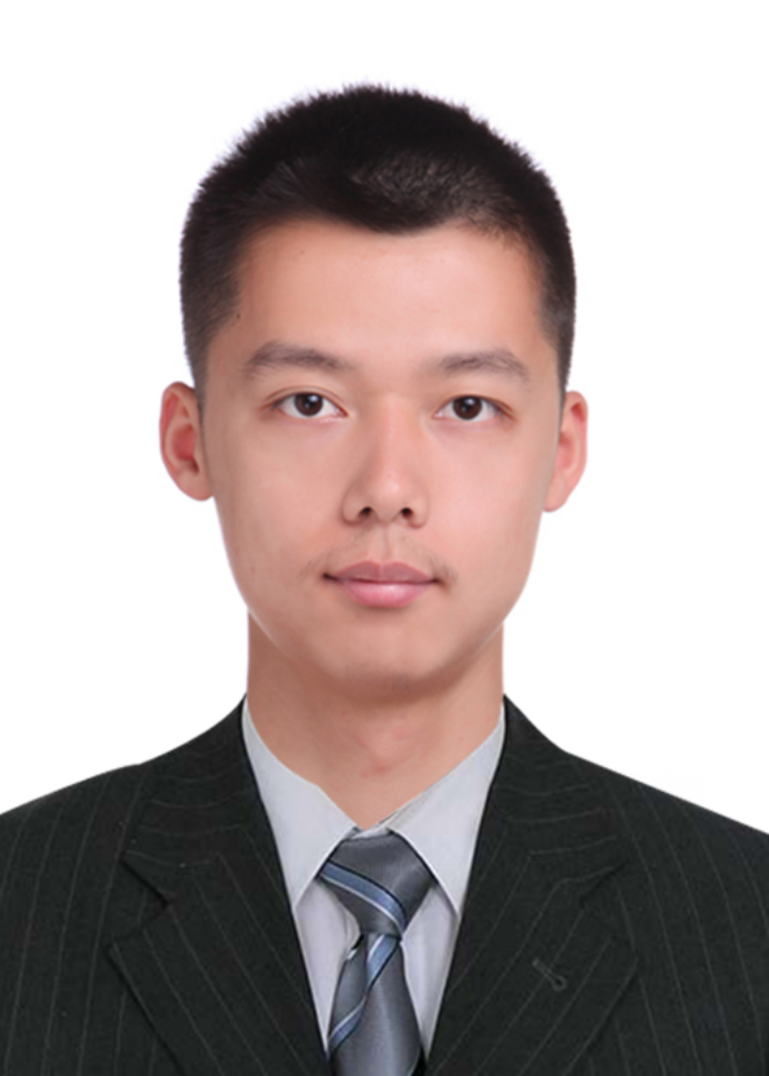}}]
	{Jiaqi Zhang} received the B.S. degree in electronic and information engineering from the School of Electronic and Information Engineering, Beijing Jiaotong University, Beijing, China, in 2016. He is currently pursuing the Ph.D. degree at the Department of Automation, Tsinghua University, Beijing, China. His research interests include networked control systems, distributed or decentralized optimization and their applications.
\end{IEEEbiography}
\begin{IEEEbiography}
	[{\includegraphics[width=1in,height=1.25in,clip,keepaspectratio]{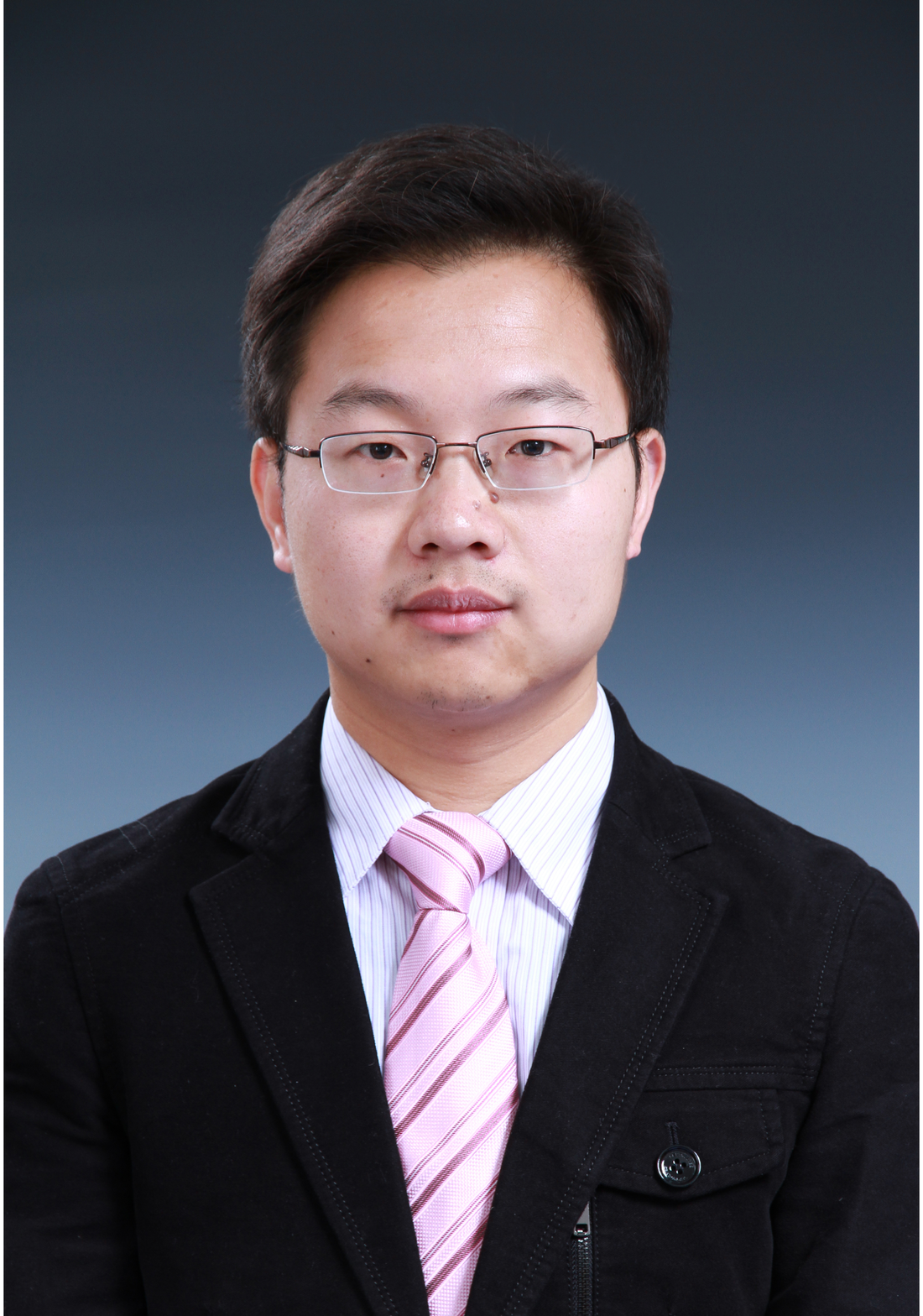}}]
	{Keyou You} (SM'17)  received the B.S. degree in Statistical Science from Sun Yat-sen University, Guangzhou, China, in 2007 and the Ph.D. degree in Electrical and Electronic Engineering from Nanyang Technological University (NTU), Singapore, in 2012. After briefly working as a Research Fellow at NTU, he joined Tsinghua University in Beijing, China where he is now a tenured Associate Professor in the Department of Automation. He held visiting positions at Politecnico di Torino,  Hong Kong University of Science and Technology,  University of Melbourne and etc. His current research interests include networked control systems, distributed optimization and learning, and their applications.

	Dr. You received the Guan Zhaozhi award at the 29th Chinese Control Conference in 2010 and the ACA (Asian Control Association) Temasek Young Educator Award in 2019. He was selected to the National 1000-Youth Talent Program of China in 2014 and received the National Science Fund for Excellent Young Scholars in 2017. He is serving as an Associate Editor for the IEEE Transactions on Cybernetics,  IEEE Control Systems Letters(L-CSS), Systems \& Control Letters.
\end{IEEEbiography}
\begin{IEEEbiography}
	[{\includegraphics[width=1in,height=1.25in,clip,keepaspectratio]{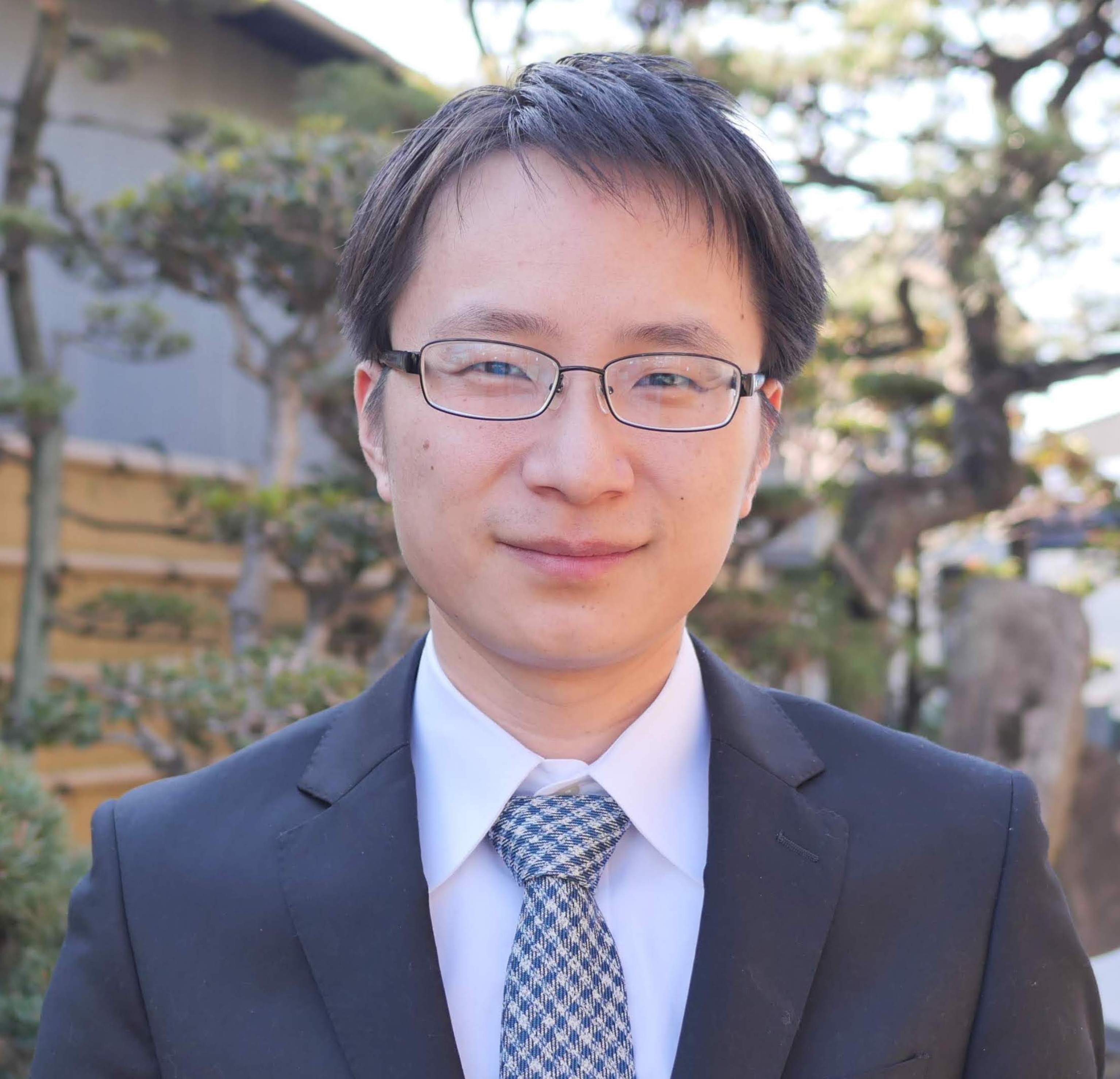}}]
	{Kai Cai} (S'08-M'12-SM'17) received the B.Eng. degree in Electrical Engineering from Zhejiang University, Hangzhou, China, in 2006; the M.A.Sc. degree in Electrical and Computer Engineering from the University of Toronto, Toronto, ON, Canada, in 2008; and the Ph.D. degree in Systems Science from the Tokyo Institute of Technology, Tokyo, Japan, in 2011. He is currently an Associate Professor at Osaka City University. Previously, he was an Assistant Professor at the University of Tokyo (2013-2014), and a postdoctoral Fellow at the University of Toronto (2011-2013).
	Dr. Cai's research interests include distributed control of discrete-event systems and cooperative control of networked multi-agent systems. He is the co-author (with W.M. Wonham) of Supervisory Control of Discrete-Event Systems (Springer 2019) and Supervisor Localization (Springer 2016). He is serving as an Associate Editor for the IEEE Transactions on Automatic Control.
\end{IEEEbiography}

\end{document}